\documentclass[acmsmall,nonacm,acmthm,screen]{acmart}

\usepackage{amsmath,amsthm}
\usepackage{multicol, makecell, multirow}
\usepackage{xspace}
\usepackage{booktabs}

\usepackage{listings}
\usepackage{algorithm,algorithmicx,algpseudocode}
\algrenewcommand\algorithmicrequire{\textbf{Input:}}
\algrenewcommand\algorithmicensure{\textbf{Output:}}

\usepackage{tikz}
\usetikzlibrary{arrows}
\usetikzlibrary{positioning}
\usetikzlibrary{fit,backgrounds}
\usetikzlibrary{decorations.shapes}
\tikzstyle{graph}=[>=stealth',semithick,inner sep=2pt,minimum width=6mm]
\tikzstyle{node}=[draw,circle,fill=white]

\newcommand{\proc}[1]{\textsc{#1}}
\newcommand{\mreach}{\rightleftharpoons}
\newcommand{\reach}{\rightharpoonup}
\newcommand{\succs}{\mathit{Successors}}
\newcommand{\preds}{\mathit{Predecessors}}
\newcommand{\preno}[1]{\mathit{Predicates}(#1)}
\newcommand{\compute}{\textnormal{\textsc{compute}}\xspace}
\newcommand{\computeVps}{\textnormal{\textsc{compute$V_p$s}}\xspace}
\newcommand{\visit}{\textnormal{\textsc{visit}}\xspace}
\newcommand{\dod}[3]{{#1}\xrightarrow{\smash{\raisebox{-2pt}{\textsf{\tiny DOD}}}}\{#2,#3\}}
\newcommand{\ntscd}[2]{{#1}\xrightarrow{\smash{\raisebox{-2pt}{\textsf{\tiny NTSCD}}}}{#2}}

\setcopyright{none}

\usepackage{booktabs}   
\usepackage{subcaption} 

\begin{document}


\title{Fast Computation of Strong Control Dependencies}


\author{Marek Chalupa}
\affiliation{
  \department{Faculty of Informatics}              
  \institution{Masaryk University}            
  \city{Brno}
  \country{Czech Republic}                    
}
\email{chalupa@fi.muni.cz}          

\author{David Klaška}
\affiliation{
  \department{Faculty of Informatics}              
  \institution{Masaryk University}            
  \city{Brno}
  \country{Czech Republic}                    
}

\author{Jan Strejček}
\orcid{0000-0001-5873-403X}
\affiliation{
  \department{Faculty of Informatics}              
  \institution{Masaryk University}            
  \city{Brno}
  \country{Czech Republic}
}
\email{strejcek@fi.muni.cz}

\author{Lukáš Tomovič}
\affiliation{
  \department{Faculty of Informatics}              
  \institution{Masaryk University}            
  \city{Brno}
  \country{Czech Republic}
}

\begin{abstract}
  We introduce new algorithms for computing non-termination sensitive
  control dependence (NTSCD) and decisive order dependence
  (DOD). These relations on control flow graph vertices have many
  applications including program slicing and compiler
  optimizations. Our algorithms are asymptotically faster than the
  current algorithms. We also show that the original algorithms for
  computing NTSCD and DOD may produce incorrect results. We
  implemented the new as well as fixed versions of the original
  algorithms for the computation of NTSCD and DOD and we
  experimentally compare their performance and outcome. Our algorithms
  dramatically outperforms the original ones.
\end{abstract}





\maketitle


\section{Introduction}

Analysis of control dependencies in programs plays an important role
in compiler optimizations~\citep{Kennedy90,Ferrante87,Cytron91,Darte00},
program analysis~\citep{Jackson94,Lechenet18a,ChalupaS19}
and program transformations,
mainly program slicing~\citep{Weiser84,Horwitz90,Amtoft08}.

Since the 70s, when the first paper about control dependence appeared,
there have been invented many different notions of control dependence.
The one that is as of today considered as the ``classical`` control dependence
was introduced by \citet{Ferrante87} more than thirty years ago.

Informally, two statements in a program are control dependent if one controls
the execution of the other in some way.
This is typically the case for \textbf{if} statements and their bodies,
but control dependence can emerge also when a loop infinitely delays
the execution of some statement.
Consider, for example, the control flow graph of a program in
Figure~\ref{fig:cd_example}.
Node $a$ controls whether $c$ or $b$ is going to be executed,
so $c$ and $b$ are control dependent on $a$ (note that the convention
is to display dependence edges in the ``controls'' direction).
Similarly, $b$ controls the execution of $c$ and $d$, as these nodes
may be bypassed by going from $b$ to $e$.
However, $c$ does not control $d$ as any path from $c$ hits $d$ eventually.

\begin{figure}[t]
\begin{tikzpicture}[graph]
\node[draw, circle] (a) at (0, 0)   {$a$};
\node[draw, circle] (b) at ( 1.5, -.5) {$b$};
\node[draw, circle] (c) at ( 1.5, .5) {$c$};
\node[draw, circle] (d) at ( 3, .5) {$d$};
\node[draw, circle] (e) at ( 4.5, 0) {$e$};

\draw[->] (a) to (b);
\draw[->] (a) to (c);
\draw[->] (c) to (d);
\draw[->] (d) to (e);
\draw[->] (b) to (c);
\draw[->] (d) edge[loop, looseness=6] (d);
\draw[->] (b) to (e);

\draw[->, bend left, red] (a) to (c);
\draw[->, bend right, red] (a) to (b);
\draw[->, bend right, red] (b) to (c);
\draw[->, red] (b) to (d);
\draw[->, bend left, dotted, red] (d) to (e);
\draw[->, bend right, red, dotted] (b) to (e);
\draw[->] (d) edge[loop, red, dotted, looseness=10] (d);
\end{tikzpicture}
\caption{An example of a control flow graph and control dependencies
         (red edges). The dotted control dependencies are non-termination sensitive.}
\label{fig:cd_example}
\end{figure}
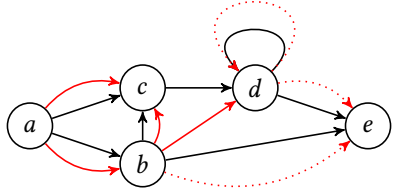

The described dependencies were those that originate
due to \textbf{if} statements.  The figure shows also dependencies
that arise due to possibly non-terminating loops (dotted red edges).
These dependencies are called \emph{non-termination sensitive}
and belong to the class of \emph{strong control dependencies}~\citep{Danicic11}.
For example, node $e$ is (non-termination sensitive) control dependent
on node $b$, because when $b$ follows the $c$ successor,
the program may get stuck in the loop over $d$ and never reach $e$.
As the calling suggests, these dependencies are sensitive to
(non-)termination of programs.

In this paper, we are concerned with the computation of strong control dependencies.
In particular, we are concerned with the computation of
\emph{non-termination sensitive control dependence (NTSCD)}
and \emph{decisive order dependence (DOD)},
both introduced by \citet{Ranganath05,Ranganath07}.

We report on flaws of the original Ranganath et al.'s algorithms
that result in incorrectly computed dependencies and suggest how to fix them.
Moreover, we provide new algorithms that are easier
and faster in theory and also in practice.
Our algorithm for the computation of NTSCD has the worst-case execution time complexity
$O(|V|^2)$ whereas the algorithm of \citet{Ranganath07} has the complexity
$O(|V|^4 \cdot log(|V|))$, where $|V|$ is the number of nodes in the
control flow graph.
Similarly, for the computation of DOD we provide an algorithm that runs
in $O(|V|^3)$ while the Ranganath et al.'s algorithm runs in $O(|V|^5 \cdot log(|V|))$.
Our algorithm is also optimal in the sense that it computes ternary relation
that can have size $O(|V|^3)$ in $O(|V|^3)$ time.

The rest of the paper is structured as follows.
At the end of this section, we summarize our contributions and
discuss the related work.
Section~\ref{sec:prelim} contains the necessary preliminaries.
Section~\ref{sec:ntscd} defines NTSCD and discusses the flaw in the
Ranganath et al.'s algorithm. Subsequently we describe our new algorithm
for the computation of NTSCD.
Similarly structured is Section~\ref{sec:dod} that is concerned with
the DOD relation. We define DOD, discuss the known algorithm for its computation
and then we develop a theory that underpins our algorithm for the computation
of DOD. The algorithm is presented at the end of the section.
In Section~\ref{sec:cc} is a discussion about control closures,
which are as of today considered as the state-of-the-art approach
to the computation of control dependencies.
Section~\ref{sec:experiments} contains experimental evaluation of our work.
The paper is concluded with Section~\ref{sec:conclusion}.

\subsection{Contributions}
The main contributions of the paper are as follows:
\begin{enumerate}
\item We have identified flaws in algorithms by \citet{Ranganath07}
      for the computation of NTSCD and DOD relations.
      We describe why the algorithms may fail and suggest a fix for each of them.
 \item We designed and implemented new algorithms for the computation
       of NTSCD and DOD that are asymptotically faster than the algorithms
       by Ranganath et al. Our algorithm for the computation of NTSCD
       runs in $O(|V|^2)$ and for the computation of DOD in $O(|V|^3)$.
       This is an improvement over the complexity of Ranganath et al.'s
       algorithms which run in $O(|V|^4 \cdot log(|V|))$ and
       $O(|V|^5 \cdot log(|V|))$, respectively.
       Moreover, our algorithm for the computation of DOD is optimal
       in the sense that it computes results of the size $O(|V|^3)$
       in $O(|V|^3)$ time.
\item We provide a solid experimental results that our algorithms are
      usable in practice and significantly faster than the previously
      known algorithms.
\end{enumerate}

\subsection{Related Work}

There are many similar yet different notions of control dependence.
The first paper concerned with control dependence is due to
\citet{Denning77} who used control dependence to
certify that flow of information in a program is secure
(as determined by security requirements given as a lattice of flow constraints).

\citet{Weiser84}, \citet{Ottenstein84}, and \citet{Ferrante87} used control
dependencies in program slicing, which is also the motivation for the most of
the latter research in this area.  These papers, generally accepted
as the ``classical'' literature about control dependence,  study algorithms
that work with procedures having a unique exit point and that always terminate.
Such restrictions do not fit modern program structures, so researchers
have been gradually extending the algorithms to more general settings.

\citet{Podgurski90} removed the restriction about termination of the program,
but still work with procedures that have a unique exit point.
\citet{Bilardi96} achieved a generalization of previously known notions of
control dependence by generalizing dominance relation in graphs.


Control dependencies are in most of the cases used at the level of standalone
procedures. However, there are cases when also information about
\emph{interprocedural} control dependencies is needed.
This direction was studied by \citet{Loyall93} and \citet{Harrold98,Sinha01}.

The notion of NTSCD and DOD was founded
in works of \citet{Ranganath05,Ranganath07} in order
to slice reactive systems. Ranaganath et al. introduced also
a \emph{non-termination insensitive control dependence}
that generalizes the classical control dependence (that assumes termination
of programs) to graphs without a unique exit point.
Further, they relax the conditions of DOD to
obtain several other order dependencies that may be useful in program slicing.

\citet{Danicic11} provided an elegant generalization of the notion of
control dependence using a so-called control closures.  In their
works, \emph{weak} control closures generalize termination
\emph{in}sensitive control dependence and \emph{strong} control
closures generalize termination sensitive (and thus also NTSCD)
control dependence.  Danicic et al. provide also algorithms for the
computation of weak and strong control closures.

\citet{Lechenet18} proved the correctness of the Danicic et al.'s weak
control closure algorithm in Coq and provided several optimizations
that significantly improve the performance of the algorithm (the
optimized algorithm have been mechanically proved in Why3).  Our
algorithms are concerned with NTSCD and DOD and thus with strong
control closures.


\section{Preliminaries}\label{sec:prelim}

The paper presents several algorithms working with finite directed
graphs. Here we recall some basic terms.

A \emph{finite directed graph} is a pair $G = (V,E)$, where $V$ is a
finite set of \emph{nodes} and $E\subseteq V\times V$ is a set of
\emph{edges}. If there is an edge $(m,n)\in E$, then $n$ is called a
\emph{successor} of $m$, $m$ is a \emph{predecessor} of $n$, and the
edge is an \emph{outgoing edge} of $m$. Given a node $n$, the set of
all its successors is denoted by $\succs(n)$ and the set of all its
predecessors is denoted by $\preds(n)$. A \emph{path} from a node
$n_1$ is a nonempty finite or infinite sequence
$n_1n_2\ldots\in V^+\cup V^\omega$ of nodes such that there is an edge
$(n_i,n_{i+1})\in E$ for each pair $n_i, n_{i+1}$ of adjacent nodes in
the sequence. A path is called \emph{maximal} if it cannot be
prolonged, i.e., if it is infinite or the last node of the path has no
outgoing edge. A node $m$ is \emph{reachable} from a node $n$, written
$n\reach m$, if there exists a finite path such that its first node is
$n$ and its last node is $m$. Further, nodes $m,n$ are \emph{mutually
  reachable}, written $m\mreach n$, whenever $m\reach n$ and
$n\reach m$. The relation $\mreach$ is clearly an equivalence on $V$
and it provides a partition of $V$ into equivalence classes called
\emph{strongly connected components (SCC)}. An SCC $S$ is called
\emph{terminal} if there is no edge leading outside this SCC, i.e.,
$E\cap(S\times (V\smallsetminus S))=\emptyset$. For each set of nodes
$V'\subseteq V$ we define a graph \emph{induced by} $V'$ as the graph
$(V',E')$ where $E'=E\cap(V'\times V')$.  An SCC is called
\emph{trivial} if it induces a graph without any edge. It is called
\emph{nontrivial} otherwise. Note that a trivial SCC has to be a
singleton. When we talk about an SCC, we usually mean the graph
induced by this SCC.

We say that a graph is a \emph{cycle}, if it is isomorphic (i.e.,
identical up to renaming nodes) to a graph $(V,E)$ where
$V=\{n_1,\ldots,n_k\}$ for some $k>0$ and
$E=\{(n_1,n_2),(n_2,n_3),\ldots,(n_{k-1},n_k),(n_k,n_1)\}$.  A cycle
\emph{unfolding} is a path in the cycle that contains each node
precisely once.

In this paper, we consider programs represented by control flow graphs,
where nodes correspond to program statements and edges model the flow
of control between the statements. As control dependence reflects only
the program structure, our definition of a control flow graph does not
contain any statements. Our definition also does not contain any start
or exit nodes as these are not important for the problems we study in
this paper.

\begin{definition}[Control flow graph]
  A \emph{control flow graph (CFG)} is a finite directed graph
  $G=(V,E)$ where each node $v\in V$ has at most two outgoing
  edges. Nodes with exactly two outgoing edges are called
  \emph{predicate nodes} or simply \emph{predicates}. The set of all
  predicates of a CFG $G$ is denoted by $\preno{G}$.
\end{definition}

Note that the number of edges in each control flow graph is at most
twice the number of its nodes. We often use this fact in the
complexity analysis of presented algorithms.




\section{Non-termination sensitive control dependence}
\label{sec:ntscd}

This section focuses on \emph{non-termination sensitive control
  dependence (NTSCD)} introduced by \citet{Ranganath05}. We start with
its formal definition.

\begin{definition}[Non-termination sensitive control dependence]\label{def:ntscd}
  Given a CFG $G=(V,E)$, a node $n\in V$ is \emph{non-termination
    sensitive control dependent} on a predicate node $p\in\preno{G}$,
  written $\ntscd{p}{n}$, if $p$ has two successors $s_1$ and $s_2$
  such that
\begin{itemize}
 \item all maximal paths from $s_1$ contain $n$, and
 \item there exists a maximal path from $s_2$ that does not contain $n$.
\end{itemize}
\label{def:ntscd}
\end{definition}

Before we present our algorithm for the computation of non-termination
sensitive control dependencies, we recall the algorithm of
\citet{Ranganath07} and show that it can produce incorrect results.

\subsection{Algorithm of \citet{Ranganath07} for NTSCD}
\label{ssec:ntscd_ranganath}

The algorithm is presented in Algorithm~\ref{fig:ranganath}. The
central data structure of the algorithm is a two-dimensional array $S$
where for each node $n$ and for each predicate node $p$ with
successors $r$ and $s$, $S[n,p]$ always contains a subset of
$\{t_{pr},t_{ps}\}$.  Intuitively, $t_{pr}$ should be added to
$S[n,p]$ if $n$ appears on all maximal paths from $p$ that start with
the prefix $pr$. The $\mathit{workbag}$ holds the set of nodes $n$ for
which some $S[n,p]$ value has been changed and this change should be
propagated.  The first part of the algorithm initializes the array $S$
with the information that each successor $r$ of a predicate node $p$
is on all maximal paths from $p$ starting with $pr$. The main part of
the algorithm then spreads the information about reachability on all
maximal paths in the forward manner. Finally, the last part computes
the NTSCD relation according to Definition~\ref{def:ntscd} and with use
of the information in $S$.

\begin{algorithm}[t!]
   \begin{algorithmic}[1]
   \Require a CFG $G = (V,E)$ 
   \Ensure a potentially incorrect NTSCD relation stored in $\mathit{ntscd}$
   \Statex 
   \State assume that $S[n,p]=\emptyset$ for all $n\in V$ and $p\in\preno{G}$ \Comment{ Initialization }
   \State $\mathit{workbag} \gets \emptyset$
   \For {$p\in \preno{G)}$} 
     \For {$r \in \succs(p)$}
       \State $S[r, p] \gets \{t_{pr}\}$
       \State $\mathit{workbag} \gets \mathit{workbag} \cup \{r\}$
     \EndFor
   \EndFor\\

   \While {$\mathit{workbag} \not = \emptyset$} \Comment { Computation of $S$ }
     \State $n \gets$ pop from $\mathit{workbag}$
     \If { $\succs(n)$ contains a single successor $s \neq n$ }
       \Comment{ \textcolor{gray}{One successor case} }
       \For { $p \in \preno{G}$ }
         \If { $S[n, p] \smallsetminus S[s, p] \not = \emptyset$ }
           \State $S[s, p] \gets S[s, p] \cup S[n, p]$
           \State $\mathit{workbag} \gets \mathit{workbag} \cup \{s\}$
         \EndIf
       \EndFor
     \EndIf

     \If { $|\succs(n)| > 1$ }
     \Comment{ \textcolor{gray}{Multiple successors case} }
     \For { $m \in V$ }
       \If { $|S[m, n]| = |\succs(n)|$}
         \For { $p \in \preno{G} \smallsetminus \{n\}$}
           \If { $S[n, p] \smallsetminus S[m, p] \not = \emptyset $}
             \State $S[m, p] \gets S[m, p] \cup S[n, p]$
             \State $\mathit{workbag} \gets \mathit{workbag} \cup \{m\}$
           \EndIf
         \EndFor
       \EndIf
     \EndFor
     \EndIf
   \EndWhile\\

   \State $\mathit{ntscd} \gets \emptyset$ \Comment { Computation of NTSCD }
   \For { $n \in V$ } 
     \For { $p \in \preno{G}$ }
       \If { $0 < |S[n, p]| < |\succs(p)|$ } 
         \State $\mathit{ntscd} \gets \mathit{ntscd} \cup \{\ntscd{p}{n}\}$ 
       \EndIf
     \EndFor
   \EndFor
\end{algorithmic}
\caption{The NTSCD algorithm by \citet{Ranganath07}}
\label{fig:ranganath}
\end{algorithm}

The complexity of the algorithm is
$O(|E|\cdot|V|^3\cdot\log|V|)$~\cite{Ranganath07} for a CFG
$G=(V,E)$. The $\log|V|$ factor comes from set operations.  Since
every node in CFG has at most 2 outgoing edges, we can simplify the
worst-case complexity to $O(|V|^4\cdot\log|V|)$.

Althought the correctness of the algorithm has been proved
\citep[Theorem~7]{Ranganath05}, we present an example where the
algorithm provides an incorrect answer. Consider the CFG in
Figure~\ref{fig:ranganath_cfg}. The first part of the algorithm
initializes the array $S$ as shown in the figure and sets
$\mathit{workbag}$ to $\{2,6,3,4\}$. Then any node from
$\mathit{workbag}$ can be popped and processed. Let us apply the
policy used for queues: always pop the oldest element in
$\mathit{workbag}$. Hence, we pop $2$ and nothing happens as the
condition on line~22 is not satisfied for any $m$. This also means
that the symbol $t_{12}$ is not propagated any further. Next we pop
$6$, which has no effect as $6$ has no successor. By processing $3$
and $4$, $t_{23}$ and $t_{24}$ are propagated to $S[5,2]$ and $5$ is
added to the $\mathit{workbag}$. Finally, we process $5$ and set
$S[6,2]$ to $\{t_{23},t_{24}\}$. The final content of $S$ is provided
in the figure. Unfortunately, the information in $S$ is sound but
incomplete.  In other words, if $t_{pr}\in S[n,p]$, then $n$ is indeed
on all maximal paths from $p$ starting with $pr$, but the opposite
implication does not hold. In particular, $t_{12}$ is missing in
$S[5,1]$ and $S[6,1]$. This flaw causes that the last part of the
algorithm computes an incorrect NTSCD relation.  More precisely, it
correctly produces dependecies $\ntscd{1}{2}$, $\ntscd{2}{3}$, and
$\ntscd{2}{4}$, but it also incorrectly produces the dependence
$\ntscd{1}{6}$ and it misses the dependence $\ntscd{1}{5}$.

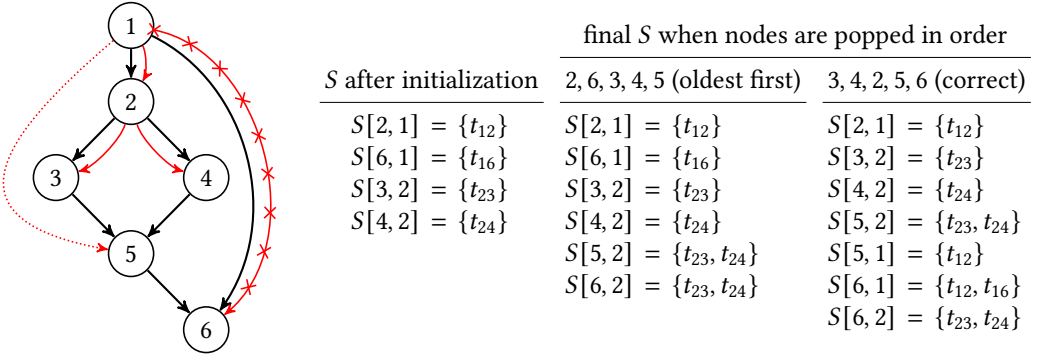
\begin{figure}[t]
\begin{center}
\begin{tikzpicture}[graph]
\node[node] (1) at (-1,  0) {1};
\node[node] (2) at (-1, -1) {2};
\node[node] (3) at (-2, -2) {3};
\node[node] (4) at (0, -2) {4};
\node[node] (5) at (-1, -3) {5};
\node[node] (6) at (0, -4) {6};
\node (7) at (-2.5,-4) {~};

\draw[->,thick] (1) to (2);
\draw[->,thick,bend left=50] (1) to (6);
\draw[->,thick] (2) to (3);
\draw[->,thick] (2) to (4);
\draw[->,thick] (3) to (5);
\draw[->,thick] (4) to (5);
\draw[->,thick] (5) to (6);

\draw[->, red, bend left , looseness=0.6] (1) to (2);
\draw[->, red, bend left , looseness=0.6] (2) to (3);
\draw[->, red, bend right, looseness=0.6] (2) to (4);
\draw[->, red, bend left=65] (1) to (6);
\draw[red, bend left=65, decorate, 
  decoration={crosses,segment length=5.15mm, shape size=1.5mm}] (1) to (6);
\draw[->, red, densely dotted, out=-140, in=170, looseness=2, overlay] (1) to (5);

\node at (6.2,-2) {
  \setlength{\arraycolsep}{2pt}
  $
   \begin{array}{rclp{.5ex}rclp{.5ex}rcl}
   &&&& \multicolumn{7}{c}{\text{final $S$ when nodes are popped in order}}
   \\ \cmidrule{5-11}
   \multicolumn{3}{c}{\text{$S$ after initialization}}
   && \multicolumn{3}{c}{\text{$2,6,3,4,5$ (oldest first)}}
   && \multicolumn{3}{c}{\text{$3,4,2,5,6$ (correct)}}
   \\ \cmidrule{1-3}\cmidrule{5-7}\cmidrule{9-11}
   \quad S[2,1]&=&\{t_{12}\} && S[2,1]&=&\{t_{12}\} && S[2,1]&=&\{t_{12}\}\\
   S[6,1]&=&\{t_{16}\} && S[6,1]&=&\{t_{16}\} && S[3,2]&=&\{t_{23}\}\\
   S[3,2]&=&\{t_{23}\} && S[3,2]&=&\{t_{23}\} && S[4,2]&=&\{t_{24}\}\\
   S[4,2]&=&\{t_{24}\} && S[4,2]&=&\{t_{24}\} && S[5,2]&=&\{t_{23},t_{24}\}\\
   &&&& S[5,2]&=&\{t_{23},t_{24}\} && S[5,1]&=&\{t_{12}\}\\
   &&&& S[6,2]&=&\{t_{23},t_{24}\} && S[6,1]&=&\{t_{12},t_{16}\}\\
   &&&&                        &&&& S[6,2]&=&\{t_{23},t_{24}\}\\
   \end{array}
  $ 
};

\end{tikzpicture}
\end{center}
\caption{An example that shows the incorrectness of the NTSCD
  algorithm by \citet{Ranganath07}. Solid red edges depict the dependencies
  computed by the algorithm when it always pops the oldest element in $\mathit{workbag}$.
  The crossed dependence is incorrect. The dotted dependence is missing in the result.}
\label{fig:ranganath_cfg}
\end{figure}

A neccessary condition to get the correct result is to process $2$
only after $3,4$ are processed and $S[5,2]=\{t_{23},t_{24}\}$. For
example, one obtains the correct $S$ (also shown in the figure) when
the nodes are processed in the order $3,4,2,5,6$.

The algorithm is clearly sensitive to the order of popping nodes from
$\mathit{workbag}$. We are currently not sure whether for each $CFG$
there exists an order that leads to the correct result. An easy way to
fix the algorithm is to ignore the $\mathit{workbag}$ and repeatedly
execute the body of the \textbf{while} loop (lines 12--31) for all
$n\in V$ until the array $S$ reaches a fixpoint. However, this
modification would slow down the algorithm substantially. Computing
the fixpoint needs $O(|V|^3)$ iterations over the loop body (lines
12--13) and one iteration of the body without the $\mathit{workbag}$
handling instructions needs $O(|V|^2)$. Hence, the overall time
complexity of the fixed version is $O(|V|^5)$.

\subsection{New algorithm for NTSCD}

We have designed and implemented an algorithm for the computation of
NTSCD which is correct, significantly simpler and asymptotically
faster than the original algorithm of \citet{Ranganath07}.

Roughly speaking, our algorithm calls for each node $n$ a procedure that identifies all
NTSCD dependencies of $n$ on predicate nodes.
%
The procedure works in the following steps.
\begin{enumerate}
\item Color $n$ red.
\item Pick an uncolored node with a positive number of successors such
  that all these successors are red and color the node red. Repeat
  this step until no such node exists.
\item For each predicate node $p$ that has a red successor and an
  uncolored one, output $\ntscd{p}{n}$.
\end{enumerate}
Unlike the Ranganath et al.'s algorithm which works in a forward manner,
our algorithm spreads the information about
reachability of $n$ on all maximal paths in the backward direction starting
from $n$.

The algorithm is presented in Algorithm~\ref{alg:algorithm}. The
discussed procedure is called $\compute(n)$ and it actually implements
only the first two steps mentioned above. In the second step, the
procedure does not search over all nodes to pick the next node to
color. Instead, it maintains the count of uncolored
successors for each node.  Once the count drops to 0 for a node,
the node is colored red and the search continues with predecessors of this node.
The third step is implemented directly in the main loop of the algorithm.

\begin{algorithm}[t!]
\begin{algorithmic}[1]
   \Require a CFG $G = (V,E)$ 
   \Ensure the NTSCD relation stored in $\mathit{ntscd}$
   \Statex 
   \Procedure{visit}{$n$} \Comment {Auxiliary procedure}
   \State $n.\mathit{counter} \gets n.\mathit{counter} - 1$
   \If {$n.\mathit{counter} = 0$}
     \State $n.\mathit{color} \gets \mathit{red}$
     \For {$m \in \preds(n)$}
       \State \visit($m$)
     \EndFor
   \EndIf
   \EndProcedure\\
   
   \Procedure{compute}{$n$} \Comment {Coloring the graph red for a given $n$}
    \For {$m \in V$} 
      \State $m.\mathit{color} \gets \mathit{uncolored}$
      \State $m.\mathit{counter} \gets |\succs(m)|$
    \EndFor\\
    \State $n.\mathit{color} \gets \mathit{red}$ 
    \For {$m \in \preds(n)$}
      \State \visit($m$)
    \EndFor
   \EndProcedure\\

  \State $\mathit{ntscd} \gets \emptyset$ \Comment { Computation of NTSCD }
  \For {$n \in V$} 
    \State \compute($n$)
    \For {$p \in \preno{G}$}
      \If {$p$ has a $red$ successor and an $\mathit{uncolored}$ successor}
        \State $\mathit{ntscd} \gets \mathit{ntscd} \cup \{\ntscd{p}{n}\}$ 
      \EndIf
    \EndFor
  \EndFor
\end{algorithmic}
\caption{The new NTSCD algorithm}
\label{alg:algorithm}
\end{algorithm}

To prove that the algorithm is correct, we basically need to show that
when $\compute(n)$ finishes, a node $m$ is red iff all maximal paths
from $m$ contain $n$. We start with a simple observation.

\begin{lemma}\label{lem:visit}
  After $\compute(n)$ finishes, a node $m$ is red if and only if $m=n$
  or $m$ has a positive number of successors and all of them are red.
\end{lemma}
\begin{proof}
  For each node $m$, the counter is initialized to the number of its
  successors and it is decreased by calls to $\visit(m)$ each time a
  successor of $m$ gets red. When the counter drops to 0 (i.e., all
  successors of the node are red), the node is colored red.
  Therefore, if $m$ is red, it got red either on line 17 and $m = n$,
  or $m\neq n$ and $m$ is red because all its successors got red (it must have a
  positive number of successors, otherwise the counter could not be 0
  after its decrement).  In the other direction, if $m = n$, it gets
  red on line 17.  If it has a positive number of successors which all
  get red, the node is colored red by the argument above.
\end{proof}

\begin{theorem}
  After $\compute(n)$ finishes, for each node $m$ it holds that $m$ is
  red if and only if all maximal paths from $m$ contain $n$.
\end{theorem}
\begin{proof}
  (``$\Longleftarrow$'') We prove this implication by contraposition.
  Assume that $m$ is an uncolored node.  Lemma~\ref{lem:visit} implies
  that each uncolored node has an uncolored successor (if it has any).
  Hence, we can construct a maximal path from $m$ containing only
  uncolored nodes simply by always going to an uncolored successor,
  either up to infinity or up to the node with no successors. This
  uncolored maximal path cannot contain $n$ which is red.

  (``$\Longrightarrow$'') For the sake of contradiction, assume that
  there is a red node $m$ and a maximal path from $m$ that does not
  contain $n$. Lemma~\ref{lem:visit} implies that all nodes on this
  path are red. If the maximal path is finite, it has to end with a
  node without any successor. Lemma~\ref{lem:visit} says that such a
  node can be red if and only if it is $n$, which is a
  contradiction. If the maximal path is infinite, it must contain a
  cycle since the graph is finite.
  Let $r$ be the node on this cycle that has been colored red as the
  first one. Let $s$ be the successor of $r$ on the cycle.  Recall
  that $r\neq n$ as the maximal path does not contain $n$. Hence, node
  $r$ could be colored red only when all its successors including $s$
  were already red. This contradicts the fact that $r$ was colored red
  as the first node on the cycle.
\end{proof}

To determine the complexity of our algorithm on a CFG $(V,E)$, we
first analyze the complexity of one run of $\compute(n)$. The lines
12--15 iterate over all nodes. The crucial observation is that the
procedure $\visit$ is called at most once for each edge $(m,m')\in E$
of the graph: to decrease the counter of $m$ when $m'$ gets
red. Hence, the procedure $\compute(n)$ runs in $O(|V|+|E|)$. This
procedure is called on line 25 for each node $n$. Finally, lines
27--29 are executed for each pair of a node $n$ and a predicate node
$p$. This gives us the overall complexity
$O((|V|+|E|)\cdot|V| + |V|^2)=O((|V|+|E|)\cdot|V|)$. Since in control
flow graphs it holds $|E|\le2|V|$, the complexity can be simplified to
$O(|V|^2)$.

Note that our algorithm is in particular suitable for on-demand computation.
Indeed, each call to $\compute(n)$ calculates exactly the control dependencies
of $n$ and therefore may be used by a client analysis as per need.


\section{Decisive order dependence}
\label{sec:dod}

During program slicing, which is the main consumer of control
dependence, there are cases when even the use of NTSCD fails to
deliver correct results.  If the CFG of a program is
\emph{irreducible}\footnote{ An \emph{irreducible} graph is a graph
  that cannot be reduced to a single node by repeating operations: 1)
  remove self-loop 2) remove a node with a single predecessors and
  reconnect its incoming edge to successors \citep{Hecht74}. }, we may
want to mark more nodes as dependent than only those that are
non-termination sensitive control dependent.

For irreducible CFGs, \citet{Ranganath07} introduced \emph{decisive
  order dependence (DOD)}. It is a relation between a predicate node
and two other nodes that captures the cases when the two nodes both
lie on all maximal paths from the predicate node, but the order in
which they are executed may change the computed values.  As an
example, see the CFG in Figure~\ref{fig:irreducible_cfg}.  Nodes $b$
and $c$ both lie on all maximal paths originating at node
$a$. Therefore, neither $b$ nor $c$ are NTSCD on $a$.  However, $a$
controls which of $b$ and $c$ is executed first.

\begin{figure}[t]
\begin{tikzpicture}[graph]
\node[draw, circle] (a) at (0, 0)   {$a$};
\node[draw, circle] (b) at (-1, -1) {$b$};
\node[draw, circle] (c) at ( 1, -1) {$c$};

\draw[->] (a) to (b);
\draw[->] (a) to (c);
\draw[->, bend left=15] (b) to (c);
\draw[->, bend left=15] (c) to (b);
\end{tikzpicture}
\caption{An example of an irreducible CFG. For such CFG, NTSCD may fail
         to deliver sensible results.}
\label{fig:irreducible_cfg}
\end{figure}
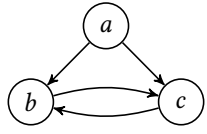

\begin{definition}[Decisive order dependence]\label{def:dod}
  Let $G=(V,E)$ be a CFG and $p,a,b\in V$ be three distinct nodes such
  that $p$ is a predicate node with successors $s_1$ and $s_2$. Nodes
  $a,b$ are \emph{decisive order-dependent (DOD)} on $p$, written
  $\dod{p}{a}{b}$, if
  \begin{itemize}
  \item all maximal paths from $p$ contain both $a$ and $b$,
  \item all maximal paths from $s_1$ contain $a$ before any occurrence
    of $b$, and
  \item all maximal paths from $s_2$ contain $b$ before any occurrence
    of $a$.
  \end{itemize}
\end{definition}

DOD is an auxiliary relation that has an effect only for irreducible
graphs. Indeed, \citet{Ranganath07} proved that DOD relation is empty
for reducible graphs.

\subsection{Algorithm of \citet{Ranganath07} for DOD}
\label{ssec:dod_ranganath}

Ranganath et al.~provide an algorithm that computes the DOD relation
for a given CFG $G=(V,E)$ in time $O(|V|^4\cdot|E|\cdot\log|V|)$ which
amounts to $O(|V|^5\cdot\log|V|)$ on CFGs~\cite[Fig.~7]{Ranganath07}.
We do not reproduce the whole algorithm here. We only mention one
unclear point.

For each triple of nodes $p,a,b\in V$ such that $p \in \preno{G}$ and
$a \neq b$, the algorithm executes the following check and if it
succeeds, then $\dod{p}{a}{b}$ is reported.
\begin{equation}
\mathit{\textsc{reachable}(a,b,G)} \land
\mathit{\textsc{reachable}(b,a,G)} \land
\mathit{\textsc{dependence}(p,a,b,G)}
\label{eq:dod}
\end{equation}
The procedure $\mathit{\textsc{dependence}}(p,a,b,G)$ returns
\emph{true} iff $a$ is on all maximal paths from one successor of $p$
before any occurrence of $b$ and $b$ is on all maximal paths from the
other successor of $p$ before any occurrence of $a$. The procedure
$\mathit{\textsc{reachable}}$ is specified only by
words~\citep[description of Fig.~7]{Ranganath07} as follows:
\begin{quote}
$\mathit{\textsc{reachable}(a,b,G)}$ returns \emph{true} if $b$ is reachable
from $a$ in the graph $G$.
\end{quote}
Unfortunately, this algorithm can provide incorrect results. For
example, consider the CFG in Figure~\ref{fig:dodcex}. For the triple
of nodes $p,a,b$, the formula~(\ref{eq:dod}) is clearly satisfied: $a$
appears on all maximal paths from one successor of $p$ (namely $a$)
before any occurrence of $b$, and $b$ appears on all maximal paths
from the other successor of $p$ (which is $b$) before any occurrence
of $a$. At the same time, $a$ and $b$ are mutually reachable.
However, it is not true that $\dod{p}{a}{b}$, because $a$ and $b$ do
not lie on all maximal paths from $p$ (the first condition of
Definition~\ref{def:dod} is violated).

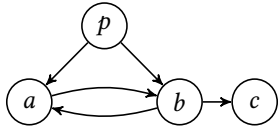
\begin{figure}[t]
\begin{tikzpicture}[graph]
\node[node] (a) at (0, 0)   {$p$};
\node[node] (b) at (-1, -1) {$a$};
\node[node] (c) at ( 1, -1) {$b$};
\node[node] (d) at ( 2, -1) {$c$};

\draw[->] (a) to (b);
\draw[->] (a) to (c);
\draw[->] (c) to (d);
\draw[->, bend left=15] (b) to (c);
\draw[->, bend left=15] (c) to (b);
\end{tikzpicture}
\caption{An example that shows incorrectness of the DOD algorithm by \citet{Ranganath07}}
\label{fig:dodcex}
\end{figure}

The algorithm can be easily fixed by modifying the procedure
$\mathit{\textsc{reachable}}(a,b,G)$ to return \emph{true} if $b$ is
on all maximal paths from $a$. This change does not increase the
overall complexity of the algorithm. Indeed, if we implement the
modified procedure $\mathit{\textsc{reachable}}(a,b,G)$ with use of
the procedure $\compute(b)$ in Algorithm~\ref{alg:algorithm}, it would
run in $O(|V|)$ time on a CFG $(V,E)$. This is exactly the complexity
of the original procedure $\mathit{\textsc{reachable}}(a,b,G)$.
By comparing the fixed and the original version of the
$\mathit{\textsc{reachable}}(a,b,G)$, one can readily confirm that the
original version produces supersets of DOD relations.

\subsection{New algorithm for DOD: crucial observations}

As in the case of NTSCD, we have designed a new algorithm for the
computation of DOD which is relatively simple and asymptotically
faster than the DOD algorithm of~\citet{Ranganath07}.

Given a CFG, our algorithm first computes for each predicate $p$ the
set $V_p$ of nodes that are on all maximal paths from $p$. Then we
process every predicate $p$ separately. We build an auxiliary graph
$A_p$ with nodes $V_p$ and from this graph we get all pairs of nodes
that are DOD on $p$.

A graph $A_p$ is a central notion of our algorithm. This subsection
highlights some important properties of these graphs and proves some
statements that underpin the algorithm. The following definitions are
needed for the construction of $A_p$.


\begin{definition}[$V'\!$-interval]
  Given a CFG $G=(V,E)$ and a subset $V' \subseteq V$, a path
  $n_1 \dots n_k$ such that $k \ge 2$, $n_1,n_k\in V'$, and
  $\forall 1 < i < k: n_i\not\in V'$ is called a
  \emph{$V'\!$-interval} from $n_1$ to $n_k$ in $G$.
 \end{definition}

In other words, a $V'\!$-interval is a finite path with at least one
edge that has the first and the last node in $V'$ but no other node on
the path is in $V'$.

\begin{definition}[Projection graph\footnote{\citet{Danicic11} call this
    graph a graph induced by V'.}]
  Given a CFG $G=(V,E)$ and a subset $V'\subseteq V$, the
  \emph{projection graph} for $V'$ is the graph $G'=(V',E')$, where
  \[
    E' = \{(x,y)\mid \text{there exists a $V'\!$-interval from $x$ to $y$ in $G$}\}.
  \]
\end{definition}

For a predicate $p$, we take the set $V_p$ of all nodes that lie on
all maximal paths from $p$ (including $p$). We compute the graph $A_p$
as the projection graph for $V_p$. Note that $A_p$ does not have to be
CFG any more as nodes in $A_p$ can have more than two
successors. However, $A_p$ represents exactly all possible orders of
the first occurrences of nodes in $V_p$ on maximal paths in $G$
starting from $p$.

\begin{lemma}\label{lem:gap}
  Let $G=(V,E)$ be a CFG and $p\in V$ be a predicate node. For each
  maximal path from $p$ in $G$, there exists a maximal path from $p$
  in $A_p$ with the same order of the first occurrences of all nodes
  from $V_p$, and vice versa.
\end{lemma}
\begin{proof}
  First, let us consider a maximal path from $p$ in $G$. The path has
  to contain all nodes from $V_p$. If there are infinitely many
  occurrences of such nodes in the path, we can easily construct a
  path in $A_p$ from $p$ which is identical to the path in $G$, only
  the nodes outside $V_p$ are omitted. Indeed, every $V_p$-interval
  $n_1\ldots n_k$ that is a part of the maximal path in $G$, can be
  simply replaced by $n_1n_k$ on the path in $A_p$ as $A_p$ contains
  the edge $(n_1,n_k)$. The situation is different when the maximal
  path in $G$ contains only finitely many occurrences of nodes from
  $V_p$. We can apply exactly the same construction on the prefix of
  this path ending by the last occurrence of a node from $V_p$. Again,
  the obtained path in $A_p$ is identical to the prefix of the path in
  $G$, only the nodes outside $V_p$ are omitted. Note that the
  obtained path in $A_p$ does not have to be maximal. However, there
  has to be a maximal path in $A_p$ with this prefix and the prefix
  has to contain all nodes from $V_p$. Hence, the order of the first
  occurrences of nodes from $V_p$ on this maximal path in $A_p$ is the
  same as on the original path in $G$.

  Second, let us consider a maximal path from $p$ in $A_p$. We can
  transform it into a path in $G$ by replacing every two successive
  nodes $m_im_{i+1}$ with a corresponding $V_p$-interval
  $m_i\ldots m_{i+1}$ in $G$. Such $V_p$-interval has to exist due to
  the construction of $A_p$. If the original maximal path in $A_p$ is
  infinite, we obtain a maximal path in $G$ that is identical to the
  original path, only some nodes outside $V_p$ are added. If the
  original maximal path in $A_p$ is finite, then the constructed path
  in $G$ does not have to be maximal, but it already contains all
  nodes from $V_p$ in the same order as the original path in $A_p$ and
  it can be prolonged into a maximal path.
\end{proof}

The definition of $V_p$ and the above lemma imply that each maximal
path from $p$ in $A_p$ contains all nodes in $A_p$.

We now turn our attention to the computation of DOD relation from $A_p$.
The following lemma solves the easy case when $p$ has zero or one
successor in $A_p$.

\begin{lemma}\label{lem:noDOD}
  If $p$ has at most one successor in $A_p$, then there are no nodes
  DOD on $p$.
\end{lemma}
\begin{proof}
  If $p$ has no successor in $A_p$, then $p$ is the only node
  contained by all maximal paths from $p$ in
  $G$. Definition~\ref{def:dod} immediately implies that there are no
  nodes DOD on $p$.

  Now assume that $p$ has a single successor $n$ in $A_p$. Let $s_1$
  and $s_2$ be the two different successors of $p$ in
  $G$. Lemma~\ref{lem:gap} implies that on all maximal paths in $G$
  from $s_1$ or $s_2$, the node $n$ appears as the first node of
  $V_p$. The paths can continue from $n$ regardless whether they start
  in $s_1$ or $s_2$. This means that if the first occurrences of nodes
  $c,d\in V_p$ appear in some order on a maximal path in $G$ from
  $s_1$, they can appear in the same order also on a maximal path from
  $s_2$. Hence, $c,d$ are not DOD on $p$.
\end{proof}

Now we study the case when $p$ has at least two different successors
in $A_p$. Actually, we start with two general lemmas not specific to
this case.

\begin{lemma}\label{lem:oneSCC}
  The graph $A_p$ has at most one nontrivial SCC. Moreover, this SCC
  is terminal.
\end{lemma}
\begin{proof}
  Recall that each maximal path from $p$ in $A_p$ has to visit all
  nodes. In particular, it has to pass through each SCC. However, when
  a path from $p$ enters its first nontrivial SCC, it can be prolonged
  into an infinite path that never leaves this SCC. Such a path is
  maximal and thus it has to contain all nodes of $A_p$. Hence, there
  is at most one nontrivial SCC in $A_p$ and it is terminal.
\end{proof}

\begin{lemma}\label{lem:sucSCC}
  All successors of $p$ in $A_p$ are in the same SCC.
\end{lemma}
\begin{proof}
  Let $n_1,n_2$ be two successors of $p$ in $A_p$. As all maximal
  paths from $p$ in $A_p$ contain all nodes, $n_2$ has to be reachable
  from $n_1$.  For the same reason, $n_1$ is reachable from
  $n_2$. Hence $n_1$ and $n_2$ are mutually reachable and thus in the
  same SCC.
\end{proof}

If $p$ has at least two successors in $A_p$, all these successors are
in the same SCC, which is therefore nontrivial. Due to
Lemma~\ref{lem:oneSCC}, this SCC is terminal. This also implies that
there is no other SCC besides the nontrivial one and potentially the
trivial SCC formed by $p$. The structure of such graph can be actually
characterized even more precisely.

\begin{lemma}\label{lem:Ap}
  If $p$ has at least two successors in $A_p$, then $p$ is a trivial SCC
  and all other nodes of $A_p$ are in a single SCC, which is a cycle.
\end{lemma}

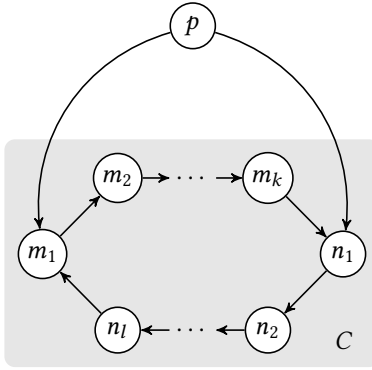
\begin{figure}[t]
  \centering
  \begin{tikzpicture}[graph]
    \node[node] (p) at (0,4) {$p$};
    \node[node] (m1) at (-2,1) {$m_1$};
    \node[node] (m2) at (-1,2) {$m_2$};
    \node (m3) at (0,2) {$\dots$};
    \node[node] (mk) at (1,2) {$m_k$};
    \node[node] (n1) at (2,1) {$n_1$};
    \node[node] (n2) at (1,0) {$n_2$};
    \node (n3) at (0,0) {$\dots$};
    \node[node] (nl) at (-1,0) {$n_l$};
    \draw[->,bend right=40] (p) to (m1);
    \draw[->,bend left=40] (p) to (n1);
    \draw[->] (m1) to (m2);
    \draw[->] (m2) to (m3);
    \draw[->] (m3) to (mk);
    \draw[->] (mk) to (n1);
    \draw[->] (n1) to (n2);
    \draw[->] (n2) to (n3);
    \draw[->] (n3) to (nl);
    \draw[->] (nl) to (m1);
    \begin{scope}[on background layer]
      \node[fill=black!10,rounded corners, inner sep=5,fit=(m1) (m2) (n1) (n2)] (rec) {};
    \end{scope}
    \node[anchor=south east, outer sep=5pt] at (rec.south east) {$C$};
  \end{tikzpicture}
  \caption{The structure of $A_p$ where $p$ has two successors}
  \label{fig:Ap}
\end{figure}

\begin{proof}
  Let $m_1$ and $n_1$ be two distinct successors of $p$ and let $C$ be
  the nontrivial SCC containing all successors of $p$. Because $C$
  contains all successors of $p$ and it is terminal, it has to contain
  all nodes of $A_p$, potentially except $p$.
  Note that $p$ is not a successor of $p$ because a
  self-loop on $p$ would generate a maximal path that does not contain
  the other successors of $p$ in $A_p$. Hence, if $p$ is not in $C$,
  it is a trivial SCC.

  As $m_1$ and $n_1$ are in one SCC, there has to be a path
  $m_1m_2\ldots m_kn_1$ from $m_1$ to $n_1$ where all nodes on the
  path are distinct.
  Analogously, there has to be a path
  $n_1n_2\ldots n_lm_1$ from $n_1$ to $m_1$ where all nodes are
  distinct. To prove that all nodes
  $m_1,m_2,\ldots,m_k,n_1,n_2,\ldots,n_l$ are distinct, it remains to
  show that $\{m_2,\ldots,m_k\}\cap\{n_2,\ldots,n_l\}=\emptyset$. For
  the sake of contradiction, assume that $m_i=n_j$ for some
  $2\le i\le k$ and $2\le j\le l$. Then there exists the infinite path
  $p(m_1m_2\ldots m_{i-1}n_jn_{j+1}\ldots n_l)^\omega$ that does not
  contain $n_1$, which is a contradiction.

  Further, it holds that
  $p\not\in\{m_1,\ldots,m_k\}\cup\{n_1,\ldots,n_l\}$. For the sake of
  contradiction, assume that $p=m_i$ (the case for $p=n_i$ is
  analogous). Then there exists the infinite path
  $p(m_1m_2\ldots m_i)^\omega$ that does not contain $n_1$, which is a
  contradiction.

  To sum up, we know that all nodes
  $p,m_1,m_2,\ldots,m_k,n_1,n_2,\ldots,n_l$ are distinct and we know
  that they are connected by edges as depicted in Figure~\ref{fig:Ap}.
  We can easily see that $A_p$ contains an infinite path
  $p(m_1m_2\ldots m_kn_1n_2\ldots n_l)^\omega$ that has to contain all
  nodes in $A_p$.

  Now we show that $p$ is not in the SCC $C$.
  For the sake of contradiction,
  assume that there is an edge $(m_i,p)$ (the case for an edge
  $(n_i,p)$ is analogous). Then there exists the infinite path
  $(pm_1m_2\ldots m_i)^\omega$ that does not contain $n_1$, which is a
  contradiction. As $p$ is not in $C$, it forms a trivial SCC.

  It remains to prove that $C$ is a cycle. In other words, we need to
  show that $C$ does not contain any other edges except those depicted
  in Figure~\ref{fig:Ap}. Intuitively, any additional edge would
  induce a maximal path from $p$ that misses some nodes of $A_p$. For
  the sake of contradiction, assume that there is an additional edge
  leading from $m_i$ (the case of an edge leading from $n_i$ is
  analogous). The edge has to match one of the following three cases.
  \begin{enumerate}
  \item $(m_i,m_j)$ where $j\le i$. Then there exists the infinite
    path $pm_1\ldots m_{j-1}(m_jm_{j+1}\ldots m_i)^\omega$ that does
    not contain $n_1$.
  \item $(m_i,m_j)$ where $j>i+1$. Then there exists the infinite path
    $p(m_1\ldots m_im_jm_{j+1}\ldots m_kn_1\ldots n_l)^\omega$ that
    does not contain $m_{i+1}$.
  \item $(m_i,n_j)$, where $i<k$ or $j>1$. Then there exists the
    infinite path $p(m_1\ldots m_in_jn_{j+1}\ldots n_l)^\omega$ that
    does not contain $m_k$ or $n_1$.
  \end{enumerate}
  In all cases, we get a contradiction.
\end{proof}

Note that Figure~\ref{fig:Ap} depicts the shape of $A_p$ where $p$ has
just 2 successors. However, $p$ may have more successors leading to
the nodes in $C$.




We now investigate how to compute the pairs of nodes $a,b$ that are
DOD on $p$ from the graph $A_p$. Lemma~\ref{lem:noDOD} says that we
can safely ignore graphs $A_p$ where $p$ has at most one successor. In
the following, we study only graphs $A_p$ where $p$ has at least two
successors. Lemma~\ref{lem:Ap} shows that these graphs have the shape
of a cycle and node $p$ with at least two outgoing edges going to some
nodes on the cycle.

The definition of DOD uses maximal paths from successors of $p$ in
$G$. Let $s_1$ and $s_2$ be the two successors of $p$ in $G$.
Note that $s_1,s_2$ can, but do not have to be in $A_p$.
For each successor $s_i$ we compute the set of
nodes that appear as the first node of $V_p$ on some maximal path from
$s_i$ in $G$. Formally, for $i\in\{1,2\}$, we define
\[
  V_i=\{n\in V_p\mid \text{there exists a path
    $s_i\ldots n\in (V\smallsetminus V_p)^*.V_p$ in $G$}\}.
\]
Note that the nodes in $V_1\cup V_2$ are exactly all the successors of
$p$ in $A_p$. The motivation for defining sets $V_i$ is that the
maximal paths from the nodes of $V_1$ in $A_p$ represent exactly all
possible orders of the first occurrences on nodes in $V_p$ on maximal
paths in $G$ starting in $s_i$. The relation is essentially the same
as the one for maximal paths from $p$ in $A_p$ and in $G$, as
described by Lemma~\ref{lem:gap}.

\begin{lemma}\label{lem:gap2}
  Let $G$ be a CFG, $p$ be a predicate node with successors
  $s_1,s_2$ in $G$ and with at least two successors in $A_p$, and let
  $i\in\{1,2\}$. For each maximal path from $s_i$ in $G$, there exists
  a maximal path from a node of $V_i$ in $A_p$ with the same order of
  the first occurrences of all nodes from $V_p\smallsetminus\{p\}$,
  and vice versa.
\end{lemma}
\begin{proof}
  Due to Lemma~\ref{lem:Ap}, the assumption that $p$ has at least two
  successors in $A_p$ means that $p$ is a trivial SCC in $A_p$. The
  construction of $A_p$ then implies that $p$ is also a trivial SCC in
  $G$. All maximal paths from $p$ in $G$ as well as in $A_p$ thus
  contain $p$ only as the starting node.
  
  Without loss of generality, assume that $i=1$. To prove the
  statement, we remove the edge $(p,s_2)$ from $G$. The statement is
  now a direct corollary of Lemma~\ref{lem:gap}, and the fact that
  maximal paths from $s_1$ in $G$ are exactly the maximal paths from
  $p$ in $G$ just without the first node, and maximal paths from nodes
  of $V_1$ in $A_p$ are exactly the maximal paths from $p$ in $A_p$
  just without the first node.
\end{proof}

Note that the sets $V_1$ and $V_2$ do not have to be disjoint. The
following lemma closes the case when they are not disjoint.

\begin{lemma}\label{lem:v1v2}
  Let $G$ be a CFG, $p$ be a predicate node with successors
  $s_1,s_2$ in $G$ and with at least two successors in $A_p$.
  If $V_1\cap V_2\neq\emptyset$, then there are no nodes DOD on $p$.
\end{lemma}
\begin{proof}
  Let $n$ be a node in $V_1\cap V_2$. Lemma~\ref{lem:gap2} implies
  that for every maximal path in $A_p$ from $n$, there exist a maximal
  path in $G$ from $s_1$ and from $s_2$ with the same order of the first
  occurrences of all nodes from $V_p\smallsetminus\{p\}$. This
  directly implies that there are no nodes DOD on $p$.
\end{proof}

The nodes in $V_i$ divide the cycle $C$ of $A_p$ into
\emph{$s_i$-strips}, which are parts of the cycle starting with a node
from $V_i$ and ending just before the next node of $V_i$.

\begin{definition}[$s_i$-strip] Let $G$ be a CFG, $p$ be a predicate
  node with successors $s_1,s_2$ in $G$ and with at least two
  successors in $A_p$, and let $i\in\{1,2\}$. An \emph{$s_i$-strip} is
  a path $n\ldots m\in V_i.(V_p\smallsetminus V_i)^*$ in $A_p$ such
  that the successor of $m$ in $A_p$ is a node in $V_i$.
\label{def:strip}
\end{definition}

\begin{figure}[t]
  \centering
  \begin{tikzpicture}[graph]
    \node[node] (p) at (0,5) {$p$};
    \node[node,fill=blue!30] (n1) at (-1.05,2.55) {$n_1$};
    \node[node,fill=red!30] (n2) at (0,3) {$n_2$};
    \node[node] (n3) at (1.05,2.55) {$n_3$};
    \node[node] (n4) at (1.5,1.5) {$n_4$};
    \node[node,fill=red!30] (n5) at (1.05,0.45) {$n_5$};
    \node[node] (n6) at (0,0) {$n_6$};
    \node[node,fill=blue!30] (n7) at (-1.05,0.45) {$n_7$};
    \node[node] (n8) at (-1.5,1.5) {$n_8$};
    \draw[->,bend right=15] (p) to (n1);
    \draw[->] (p) to (n2);
    \draw[->,out=-20,in=20,looseness=1.4] (p) to (n5);
    \draw[->,out=-160,in=160,looseness=1.4] (p) to (n7);
    \draw[->,bend left=12] (n1) to (n2);
    \draw[->,bend left=12] (n2) to (n3);
    \draw[->,bend left=12] (n3) to (n4);
    \draw[->,bend left=12] (n4) to (n5);
    \draw[->,bend left=12] (n5) to (n6);
    \draw[->,bend left=12] (n6) to (n7);
    \draw[->,bend left=12] (n7) to (n8);
    \draw[->,bend left=12] (n8) to (n1);
    \begin{scope}[on background layer]
      \draw[line width=11pt,red!15,line cap=round,bend left=20]
        (n2.south) to (n3.south west) to (n4.west);
      \draw[line width=11pt,red!15,line cap=round,bend left=20]
        (n5.north west) to (n6.north) to (n7.north east) to (n8.east) to (n1.south east);
      \draw[line width=12pt,blue!15,line cap=round,bend left=20]
        (n7.south west) to (n8.west);
      \draw[line width=12pt,blue!15,line cap=round,bend left=20]
        (n1.north west) to (n2.north) to (n3.north east) to (n4.east) to
        (n5.south east) to (n6.south);
    \end{scope}
  \end{tikzpicture}
  \caption{An example of $A_p$ with $V_1=\{n_1,n_7\}$ (blue nodes) and
    $V_2=\{n_2,n_5\}$ (marked as red nodes), $s_1$-strips
    $n_1n_2n_3n_4n_5n_6$ and $n_7n_8$ (blue strips), and $s_2$-strips
    $n_2n_3n_4$ and $n_5n_6n_7n_8n_1$ (red strips)}
  \label{fig:strips}
\end{figure}
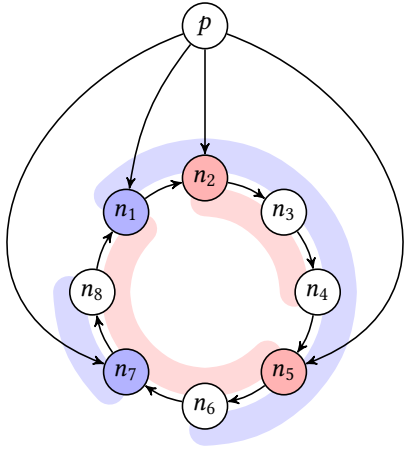

An example of $A_p$ with $s_i$-strips is in
Figure~\ref{fig:strips}. The $s_i$-strips directly say which pairs of
nodes of $V_p$ are in the same order on all maximal paths from $s_i$
in $G$.

\begin{lemma}\label{lem:strips}
  Let $G$ be a CFG, $p$ be a predicate node with successors $s_1,s_2$
  in $G$ and with at least two successors in $A_p$, and let
  $i\in\{1,2\}$. Further, let $a,b\in V_p\smallsetminus\{p\}$ be two
  distinct nodes. The node $a$ is before any occurrence of $b$ on all
  maximal paths from $s_i$ in $G$ if and only if there is an
  $s_i$-strip containing both $a$ and $b$ where $a$ is before $b$.
\end{lemma}
\begin{proof}
  (``$\Longleftarrow$'') Let $a,b$ be two nodes such that $a$ appears
  before $b$ in an $s_i$-strip. It is obvious that the first
  occurrence of $a$ appears before the first occurrence of $b$ on
  every maximal path from nodes of $V_i$ in $A_p$. Due to
  Lemma~\ref{lem:gap2}, the first occurrences of $a$ and $b$ have the
  same order also in all maximal paths from $s_i$ in $G$.
  
  (``$\Longrightarrow$'') We prove this implication by
  contraposition. The case when $b$ is before $a$ in an $s_i$-strip is
  actually proven by the previous paragraph where one just swaps $a$
  and $b$. To prove the remaining case, assume that $a$ and $b$ are in
  different $s_i$-strips. Let $n\in V_i$ be the first node of the
  $s_i$-strip containing $b$. Then the maximal path in $A_p$ from $n$
  contains $b$ before the first occurrence of $a$. Due to
  Lemma~\ref{lem:gap2}, there exists a maximal path in $G$ from $s_i$
  where $b$ is before any occurrence of $a$.
\end{proof}

The following theorem is an immediate corollary of
Lemmas~\ref{lem:noDOD} and~\ref{lem:strips}.
\begin{theorem}\label{thm:DOD}
  Let $G=(V,E)$ be a CFG and $p,a,b\in V$ be three distinct nodes such
  that $p$ is a predicate with successors $s_1$ and $s_2$. Then
  $a,b$ are DOD on $p$
  if and only if
  \begin{itemize}
  \item $p$ has at least two successors in $A_p$, 
  \item there exists an $s_1$-strip in $A_p$ that contains $a$ before $b$, and
  \item there exists an $s_2$-strip in $A_p$ that contains $b$ before $a$.
  \end{itemize}
\end{theorem}

Let us consider again the $A_p$ in Figure~\ref{fig:strips}. The
theorem implies that
nodes $n_1,n_5$ are DOD on $p$ as they appear in $s_1$-strip
$n_1n_2n_3n_4n_5n_6$ and in $s_2$-strip $n_5n_6n_7n_8n_1$ in the
opposite order. Nodes $n_1,n_6$ are DOD on $p$ for the same reason.


With use of the previous theorem, we can find a regular language over
$V_p$ such that there exist nodes $a,b$ DOD on $p$ iff an unfolding
of the cycle in $A_p$ is in the language.

\begin{theorem}\label{thm:DOD2}
  Let $G$ be a CFG, $p$ be a predicate node with successors $s_1,s_2$
  in $G$ and with at least two successors in $A_p$. Further, let
  $V_1\cap V_2=\emptyset$ and $U=V_p\smallsetminus(V_1\cup V_2)$.
  There are some nodes $a,b$ DOD on $p$ if and only if the cycle in
  $A_p$ has an unfolding of the form
  $V_1.U^*.(V_2.U^*)^*.V_2.U^*.(V_1.U^*)^*$
\end{theorem}
\begin{proof}
  ($\Longrightarrow$) Let $\dod{p}{a}{b}$ for some $a,b$.
  Theorem~\ref{thm:DOD} implies that there is an $s_1$-strip
  $m_1\ldots m_k$ that contains $a$ before $b$ and an $s_2$-strip that
  contains $b$ before $a$. Hence, the $s_2$-strip has to have the form
  $m_j\ldots m_ko_1\ldots o_lm_1\ldots m_i$ for some $1\le i<j\le k$
  such that $a$ is contained in $m_1\ldots m_i$ and $b$ in
  $m_j\ldots m_k$. The path $m_1\ldots m_ko_1\ldots o_l$ is the
  unfolding of the cycle in $A_p$ from $m_1$. Every $s_1$-strip is in
  $V_1.(U\cup V_2)^*=V_1.U^*.(V_2.U^*)^*$. Similarly, every
  $s_2$-strip is in $V_2.U^*.(V_1.U^*)^*$. Altogether, we get that the
  unfolding from $m_1$ satisfies
  $m_1\ldots m_ko_1\ldots o_l\in
  V_1.U^*.(V_2.U^*)^*.V_2.U^*.(V_1.U^*)^*$.

  ($\Longleftarrow$) Let
  $m_1\ldots m_t\in V_1.U^*.(V_2.U^*)^*.V_2.U^*.(V_1.U^*)^*$ be an
  unfolding of the cycle in $A_p$. Let $m_j$ be the node of $V_2$ in
  the unfolding. Clearly, $m_1$ is before $m_j$ in the $s_1$-strip
  starting with $m_1$ and $m_j$ is before $m_1$ in the $s_2$-string
  starting in $m_j$. Theorem~\ref{thm:DOD} implies that
  $\dod{p}{m_1}{m_j}$.
\end{proof}  

Finally, an unfolding of the correct shape can be used for direct
computation of nodes that are DOD on $p$.

\begin{theorem}\label{thm:DOD3}
  Let $G$ be a CFG, $p$ be a predicate node with successors $s_1,s_2$
  in $G$ and with at least two successors in $A_p$. Further, let
  $V_1\cap V_2=\emptyset$ and $U=V_p\smallsetminus(V_1\cup V_2)$ and
  let $A_p$ have an unfolding of the form
  $V_1.U^*.(V_2.U^*)^*.V_2.U^*.(V_1.U^*)^*$. Then there is exactly one
  path $m_1\ldots m_i\in V_1.U^*.V_2$ and exactly one path
  $o_1\ldots o_j\in V_2.U^*.V_1$ on the cycle. Moreover,
  $\dod{p}{a}{b}$ if and only if $m_1\ldots m_{i-1}$ contains $a$ and
  $o_1\ldots o_{j-1}$ contains $b$ (or the other way).
\end{theorem}
\begin{proof}
  Assume that the cycle has an unfolding of the form
  $V_1.U^*.(V_2.U^*)^*.V_2.U^*.(V_1.U^*)^*$. The fact that the cycle
  has exactly one path $m_1\ldots m_i\in V_1.U^*.V_2$ and one path
  $o_1\ldots o_j\in V_2.U^*.V_1$ is obvious.

  The equivalence follows from this picture.
    \begin{center}
    \setlength{\tabcolsep}{0pt}
    \begin{tabular}{ccccccc}
      \multicolumn{5}{c}{$s_1$-strip} 
      \\ \cmidrule[2pt]{1-5}
      $\boxed{V_1.U^*}$ & $~.~$ & $(V_2.U^*)^*$ & $~.~$ & $\boxed{V_2.U^*}$
      & $~.~$ & $(V_1.U^*)^*$
      \\[2.5pt] \cmidrule[2pt]{1-1} \cmidrule[2pt]{5-7}
      \smash{$s_2$-strip} &&&& \multicolumn{3}{c}{$s_2$-strip} \\
      \smash{(part II)}  &&&& \multicolumn{3}{c}{(part I)}
    \end{tabular}
  \end{center}
  The path $m_1\ldots m_{i-1}$ corresponds to the left box in the
  unfolding and the path $o_1\ldots o_{j-1}$ to the right box. Each
  $a$ of the path $m_1\ldots m_{i-1}$ is before each $b$ of path
  $o_1\ldots o_{j-1}$ in the marked $s_1$-strip, and they are in the
  opposite order in the marked $s_2$-strip. Due to
  Theorem~\ref{thm:DOD}, we have $\dod{p}{a}{b}$. There is no other
  DOD dependence induced by the pair of marked strips. As all other
  $s_1$-strips are subparts of the marked $s_2$-strip, they cannot
  induce any DOD dependence on $p$. Symmetric argument works for all
  other $s_2$-strips.
\end{proof}

\subsection{New algorithm for DOD: pseudocode and complexity}

The algorithm is shown in Algorithms~\ref{alg:vps}
and~\ref{alg:dodntscd}. Because nearly all applications of DOD need
also NTSCD, we presented the DOD algorithm with a simple extension
(marked with gray lines) that computes NTSCD simultaneously with DOD.

\begin{algorithm}[t!]
\algrenewcommand\algorithmicrequire{\textbf{Require:}}
\begin{algorithmic}[1]
   \Require a CFG $G = (V,E)$ 
   \Ensure sets $V_n=\{m\in V\mid m\text{ is on all maximal paths from }n\}$ for all $n\in V$ 
   \Statex
   \Procedure{visit}{$n$, $r$} \Comment {Auxiliary procedure}
   \State $n.\mathit{counter} \gets n.\mathit{counter} - 1$
   \If {$n.\mathit{counter} = 0$}
     \State $V_n \gets V_n \cup \{r\}$
     \For {$m \in \preds(n)$}
       \State \visit($m$, $r$)
     \EndFor
   \EndIf
   \EndProcedure\\

   \Procedure{compute}{$n$} \Comment {``Coloring the graph red'' for a given $n$}
    \For {$m \in V$} 
      \State $m.\mathit{counter} \gets |\succs(m)|$
    \EndFor\\
    \State $V_n \gets V_n \cup \{n\}$
    \For {$m \in \preds(n)$}
      \State \visit($m$, $n$)
    \EndFor
   \EndProcedure\\

   \Procedure{compute$V_p$s}{} \Comment {The main procedure}
   \For {$n \in V$}
     \State $V_n \gets \emptyset$
   \EndFor
   \For {$n \in V$}
     \State \compute($n$)
   \EndFor
   \EndProcedure
\end{algorithmic}
\caption{
         The algorithm computing $V_n$ for all nodes $n$
         }
\label{alg:vps}
\end{algorithm}

The DOD algoritm starts at the bottom of Algorithm~\ref{alg:dodntscd}.
The first step is to compute the sets $V_p$ for all predicate nodes
$p$ of a given CFG $G$. The computation of predicate nodes can be
found in Algorithms~\ref{alg:vps}. It is a slightly modified version
of Algorithm~\ref{alg:algorithm}. Recall that the procedure
$\compute(n)$ of Algorithm~\ref{alg:algorithm} marks red every node
such that all maximal paths from the node contain $n$. The procedure
$\compute(n)$ of Algorithm~\ref{alg:vps} does in principle the same,
but instead of the red color it marks the nodes with the identifier of
the node $n$. Every node $m$ collects these marks in set $V_m$. After
we run $\compute(n)$ for all the nodes $n$ in the graph, each node $m$
has in its set $V_m$ precisely all nodes that are on all maximal paths
from $m$. For the computation of DOD, only the sets $V_p$ for
predicate nodes $p$ are needed, but the extension computing NTSCD
may actually use all these sets.

\begin{algorithm}[t]
  \begin{algorithmic}[1]
    \Require a CFG $G = (V,E)$
    \Ensure the DOD relation stored in $\mathit{dod}$
    \textcolor{gray}{and the NTSCD relationstored in $\mathit{ntscd}$}
    \Statex
    \Procedure{computeDOD}{$p$} \Comment {Computation of DOD for predicate $p$}
    \State $A_p \gets$ \proc{build}$A_p(p)$ \Comment{ Get the graph $A_p$}
    \If { $p$ has at most one successor in $A_p$ } \Comment{ Apply Lemma~\ref{lem:noDOD}}
    \State \Return $\emptyset$
    \EndIf
    \\
    \State $V_1,V_2 \gets \text{\proc{compute}}V_1V_2(p)$
      \Comment {Get the sets $V_1,V_2$ of nodes on the cycle of $A_p$}
    \If {$V_1\cap V_2\neq\emptyset$} \Comment {Apply Lemma~\ref{lem:v1v2}}
    \State \Return $\emptyset$
    \EndIf
    \\
    \State $n_1n_2\ldots n_t \gets \text{\proc{unfoldCycle}}(A_p,V_1)$
      \Comment {Unfold the cycle of $A_p$ from a node in $V_1$}
    \State $U \gets V_p\smallsetminus(V_1\cup V_2)$
    \If {$n_1n_2\ldots n_t\not\in(V_1.U^*)^*.V_1.U^*.(V_2.U^*)^*.V_2.U^*.(V_1.U^*)^*$}
      \Comment {Apply Theorem~\ref{thm:DOD2}}
    \State \Return $\emptyset$
    \EndIf
    \\
    \State $m_1\ldots m_i \gets \text{\proc{extract}}(n_1n_2\ldots n_t,V_1.U^*.V_2)$
      \Comment {Apply Theorem~\ref{thm:DOD3}} 
    \State $o_1\ldots o_j \gets \text{\proc{extract}}(n_1n_2\ldots n_t,V_2.U^*.V_1)$
    \State \Return $\{\dod{p}{a}{b} \mid a \in \{m_1,\ldots,m_{i-1}\},
      b \in \{o_1,\ldots,o_{j-1}\}\}$
    \EndProcedure
    \\      
    \color{gray}\Procedure{computeNTSCD}{$p$}
      \Comment {Computation of NTSCD for predicate $p$}
    \State $\{s_1, s_2\} \gets \succs(p)$
    \State \Return $\{\ntscd{p}{n}\mid n\in(V_{s_1}\smallsetminus V_{s_2})
      \cup (V_{s_2}\smallsetminus V_{s_1})\}$
    \EndProcedure
    \color{black}
    \\
    \\ \Comment {\textbf{Start of the algorithm}} 
    \State \computeVps \Comment {Get $V_p$ for all predicates $p$}
    \State $\mathit{dod} \gets \emptyset$
    \State \textcolor{gray}{$\mathit{ntscd} \gets \emptyset$}
    \For {$p \in \preno{G}$}
      \Comment{ Computation of DOD \textcolor{gray}{and NTSCD}}
    \State $\mathit{dod} \gets \mathit{dod} \cup \text{\proc{computeDOD}}(p)$
    \State \textcolor{gray}{$\mathit{ntscd} \gets \mathit{ntscd} \cup
      \text{\proc{computeNTSCD}}(p)$}
    \EndFor
  \end{algorithmic}
  \caption{The new DOD algorithm. It computes also NTSCD if the gray
    lines are included.  The procedure $\computeVps$ is given in
    Algorithm~\ref{alg:vps}.}
  \label{alg:dodntscd}
\end{algorithm}

After we calculate sets $V_p$, we compute DOD and potentially also
NTSCD dependencies for each predicate node separately by procedures
\proc{computeDOD}$(p)$ and \proc{computeNTSCD}$(p)$. The procedure
\proc{computeDOD}$(p)$ first constructs the projection graph $A_p$
with use of \proc{build}$A_p(p)$. Nodes of the graph are these of
$V_p$. To compute edges, we trigger a depth-first search in $G$ from
each $n\in V_p$. If we find a node $m \in V_p$, we add the edge
$(n,m)$ to the graph $A_p$ and stop the search on this path.  When the
graph $A_p$ is constructed, we check its structure.  If the node $p$
has at most one successor in $A_p$, we return the empty relation due
to Lemma~\ref{lem:noDOD}.

The next step is to compute the sets $V_1$ and $V_2$ of nodes on
$V_p$ that are first-reachable from the two successors $s_1,s_2$ of $p$ in
$G$ via nodes outside $V_p$. Again, we apply a similar depth-first
search as in the construction of $A_p$ described above. If the sets
$V_1,V_2$ are not disjoint, we return the empty relation due to
Lemma~\ref{lem:v1v2}.

Then we unfold the cycle in $A_p$ from an arbitrary node in $V_1$,
compute the set $U$, and return the empty relation if the unfolding
does not match the pattern of Theorem~\ref{thm:DOD2}. We actually use
a slightly different pattern. But the check is correct as a cycle has
an unfolding of the form $V_1.U^*.(V_2.U^*)^*.V_2.U^*.(V_1.U^*)^*$ if
and only if all unfoldings of the cycle from $V_1$ have the form
$(V_1.U^*)^*.V_1.U^*.(V_2.U^*)^*.V_2.U^*.(V_1.U^*)^*$.

Finally, we extract the paths of the form $V_1.U^*.V_2$ and
$V_2.U^*.V_1$ from the unfolding. Note that the last node of the
latter path can be in fact the first node of the unfolding.  Finally,
we compute the DOD dependencies according to Theorem~\ref{thm:DOD3}.

The procedure \proc{computeNTSCD}$(p)$ used for the computation of
NTSCD simply follows Definition~\ref{def:ntscd}: it makes dependent on
$p$ each node that is on all maximal paths from the successor $s_1$
but not on all maximal paths from the successor $s_2$ or symmetrically
for $s_2$ and $s_1$.

As the correctness of our algorithm comes directly from the
observations made in the previous subsection, it remains only to
analyze its complexity. The procedure \proc{compute$V_p$s} consists of
two cycles in sequence. The first cycle runs in $O(|V|)$. The second
cycle calls $O(|V|)$-times the procedure $\compute(n)$. This procedure
is essentially identical to the procedure of the same name in
Algorithm~\ref{alg:algorithm} and so is its time complexity, namely
$O(|V|+|E|)$. Overall, the procedure \proc{compute$V_p$s} runs in
$O(|V|\cdot(|V|+|E|))$, which is $O(|V|^2)$ for control flow graphs.

Now we discuss the complexity of the procedure \proc{computeDOD}$(p)$.
Creating the $A_p$ graph requires calling depth-first search $O(|V|)$
times, which is $O(|V|\times|E|)$ in total. Computation of $V_1,V_2$
requires another two calls of depth-first search, which is in
$O(|E|)$. Checking that $V_1$ and $V_2$ are disjoint is in $O(|V|)$,
And so is unfolding a cycle, check of the unfolding for the pattern on
line~14, and the calls to \proc{extract}. The most demanding part is
actually the construction of the DOD relation on line~20, which takes
$O(|V|^2)$. Altogether, \proc{computeDOD}$(p)$ runs in
$O(|V|\times |E| + |V|^2)$ which simplifies to $O(|V|^2)$ for control
flow graphs.

\proc{computeDOD} is called $O(|V|)$ times, so the overall complexity
of computing DOD for a CFG $G=(V,E)$ is $O(|V|^3)$.  If we compute
also NTSCD, we make $O(|V|)$ extra calls to \proc{computeNTSCD}$(p)$,
where one call takes $O(|V|)$ time. Therefore, the asymptotic
complexity of computing NTSCD with DOD does not change from computing
DOD only.

DOD is ternary relation, therefore its maximal size is $O(|V|^3)$.
There also exists a graph for which the DOD relation reaches this limit
(a cycle on $|V|/2$ nodes together with $|V|/2$ predicate nodes; each
predicate node has the same successors on the cycle such that these
successors divide the cycle into two strips of length $|V|/4$).
Hence, our algorithm running in time $O(|V|^3)$ is optimal.


\section{Comparision to control closures}
\label{sec:cc}

In 2011, \citet{Danicic11} introduced the notion of \emph{control
  closures (CC)}.  Control closures generalize control dependence from
control flow graphs to arbitrary graphs. In particular, \emph{strong
  control closure}, which is sensitive to non-termination, generalizes
strong forms of control dependence including NTSCD and DOD.

\begin{definition}[Strongly control-closed set]
  Let $G = (V, E)$ be a CFG and let $V' \subseteq V$.  The set $V'$ is
  \emph{strongly control-closed}\footnote{The definition has been
    adjusted to our settings, notably that predicates in our CFGs
    always have two outgoing edges (our CFGs are \emph{complete} in
    the terms of \citet{Danicic11}).} in $G$ if and only if for every
  node $v\in V \smallsetminus V'$ that is reachable in $G$ from a node
  of $V'$, one of these holds:
  \begin{itemize}
  \item any path starting at $v$ avoids $V'$ or
  \item all maximal paths from $v$ contain a node from $V'$ and the
    set $\Theta(v, V')$ of the so-called \emph{first-reachable
      elements} in $V'$ from $v$ defined as 
    $$\Theta(v, V') = \{y \mid v\ldots y \in (V \smallsetminus V')^+.V'
    \textrm{is a path in }G\}$$ has at most one element, i.e.,
    $|\Theta(v, V')| \le 1$.
  \end{itemize}
\end{definition}

\begin{definition}[Strong control closure]
  Let $G = (V, E)$ be a CFG and let $W\subseteq V$. A \emph{strong
    control closure} of $W$ is a strongly control-closed set
  $V'\subseteq W$ such that there is no strongly control-closed set
  $V''$ satisfying $W\subseteq V''\subsetneq V'$.
\end{definition}


Danicic et al.~present an algorithm for the computation of strong
control closures running in $O(|V|^4)$~\citep[Theorem~66]{Danicic11}.
In fact, the algorithm uses a procedure $\Gamma$ that is very similar
to our procedure $\compute(n)$ of Algorithm~\ref{alg:algorithm}.

We can also define when a set is closed under NTSCD or DOD.
\begin{definition}[Sets closed under NTSCD/DOD]
Let $G = (V, E)$ be a CFG and let $V' \subseteq V$. 
\begin{itemize}
\item $V'$ is \emph{closed under NTSCD} if for each $n\in V'$ it holds that
  $\ntscd{p}{n}$ implies $p\in V'$.
\item $V'$ is \emph{closed under DOD} if for each $a,b\in V'$ it holds that
  $\dod{p}{a}{b}$ implies $p\in V'$.
\end{itemize}
\end{definition}
NTSCD and DOD closures are defined analogously to the strong control
closure. An NTSCD closure for a set $W\subseteq V$ can be computed by
gathering nodes backward reachable from $W$ via the NTSCD edges. The
DOD closure can be computed similarly. 

\citet[Lemmas~93 and 94]{Danicic11} proved that for a $G = (V,E)$ that
has a distinguished $\mathit{start}$ node from which all nodes in $V$
are reachable and a subset $V' \subseteq V$ such that
$\mathit{start} \in V'$, the set $V'$ is strongly control-closed if
and only if it is closed under NTSCD and DOD. Hence, the strong
control closures of the sets $W$ containing the $\mathit{start}$ node
can be computed also by computing NTSCD and DOD relations and the
backward reachability along the control dependencies. The complexity
of computing the strong control closures via NTSCD and DOD closures is
$O(|V|^3)$ as the backward reachability is dominated by the
computation of NTSCD and DOD relations.



A substantial difference between strong control closures and our
algorithms is that we are able to compute DOD and NTSCD separately,
whereas strong control closures are not. Moreover, our algorithm for
the computation of DOD and NTSCD closure is asymptotically faster.



\section{Experimental Evaluation}
\label{sec:experiments}

\subsection{Implementation}

We implemented our algorithms for the computation of NTSCD and DOD
in C++ programming language on top of the LLVM~\citep{Lattner04} infrastructure.
The implementation is now a part of a well maintained library for program analysis.
We also implemented the original Ranganath et al.'s algorithms,
the fixed versions of these algorithms as described in Subsection%
~\ref{ssec:ntscd_ranganath} and Subsection~\ref{ssec:dod_ranganath},
and the algorithm for the computation of strong control closures.

There are several minor differences in the implementation of our algorithms
compared to the description of algorithms in the text.
The main difference is that we do not use recursion.
We have rewritten the algorithms to use an explicit stack in order to avoid
stack overflow on huge graphs. This step was rather a safeguard,
as we had not hit stack overflow during any experiment with the preliminary
versions of the algorithms that used recursion.
Still, we believe that the exposition with recursion is more clear,
therefore we describe the algorithms without the explicit stack.

Another, rather minor, difference is that we use sparse bitvectors instead
of bitvectors of a fixed length. This step only saves memory on sparse graphs.

In the implementation of control closures, we use our procedure $\compute$
to implement the function $\Gamma$. This should have only a positive effect
as this procedure is more efficient than iterating over all edges
in the copy of the graph and removing them~\citep{Danicic11}.

\subsection{Evaluation}

\begin{figure}[t]
\begin{tabular}{c p{1mm} c}
\hspace*{-3mm}
\includegraphics[width=6.5cm]{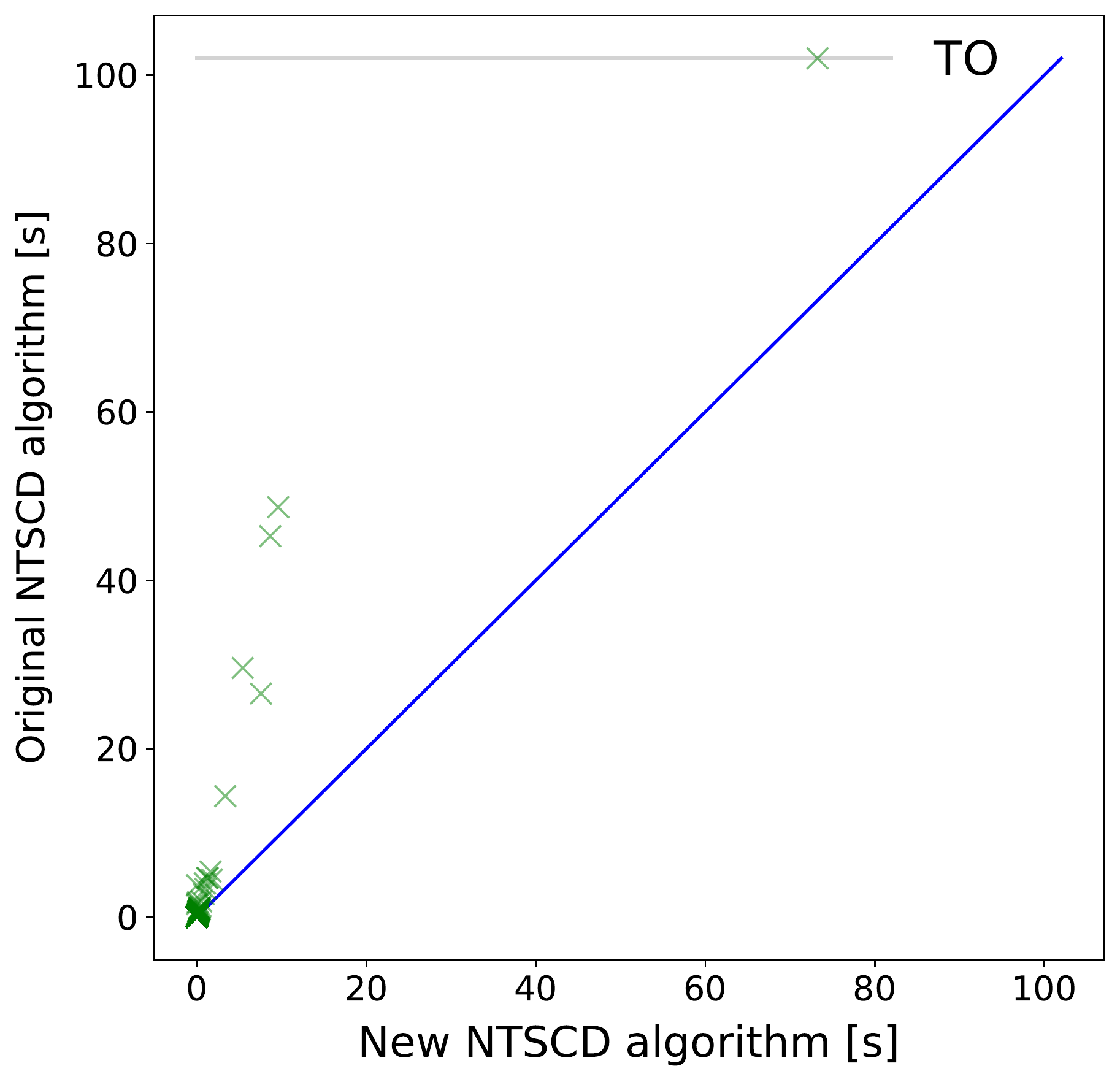}
&&
\includegraphics[width=6.5cm]{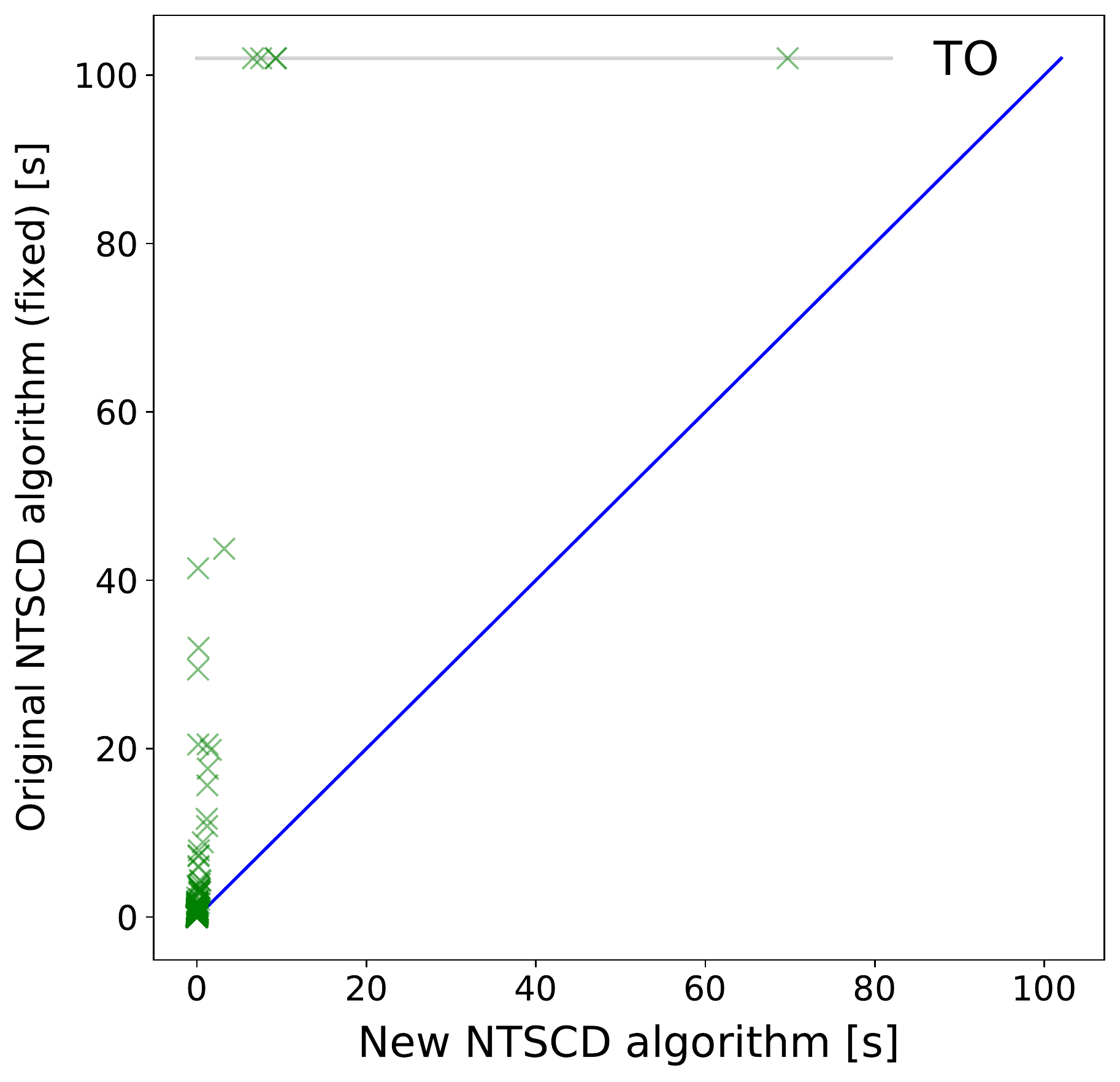}
\end{tabular}
\caption{The comparison of the running CPU time of the new NTSCD algorithm
         and the original NTSCD algorithms.
         The scatter plot on the left shows the comparison with the incorrect
         version of the original algorithm and the scatter plot
         on the right shows the comparison with the fixed version of the
         original algorithm.}
\label{fig:ntscd_comparison}
\end{figure}

We conducted a set of experiments on functions that we extracted
from SV-COMP benchmarks%
~\footnote{\url{https://github.com/sosy-lab/sv-benchmarks}}.
These benchmarks contain many artificial or generated code,
but also a lot of real-life code, e.g., from the Linux project.

Each source code file was compiled with \textsc{clang} into LLVM
and preprocessed by the \texttt{-lowerswitch} pass to ensure that every
basic block has at most two successors. Then we extracted individual
functions.
As the computation of control dependencies runs swiftly on small functions,
we filtered the extracted functions and kept only those that contain at
least 100 basic blocks. After the deletion of duplicate functions,
we were left with 2440 functions that we used for our experiments.
Run on each function comprises building a directed graph from the basic
blocks of the function and running the given algorithm on this graph.

The experiments were run on machines with \emph{AMD EPYC} CPU with the frequency
3.1\,GHz.  Each benchmark run was constrained to 1 core and 8\,GB of RAM.
For enforcing the resources isolation, we used the tool \emph{Benchexec}%
~\citep{Beyer19a}.
We set the timeout to 100\,s for each run of an algorithm.

In the following text, we denote the Ranganath et al.' algorithms
as the \emph{original} algorithms (we distinguish
between the incorrect and the fixed version at appropriate places)
and the algorithms from this paper as the \emph{new} algorithms.

\subsubsection{NTSCD algorithms}

In the first set of experiments, we compared the new NTSCD algorithm
against the original NTSCD algorithm.
We compared to both, the incorrect version
of the original algorithm and against the fixed version.
Although it seems that comparing to the incorrect version
is meaningless, we did not want to compare only to the fixed
version as the fix that we provided slows down the algorithm.

The results are depicted in Figure~\ref{fig:ntscd_comparison}.
On the left scatter plot, there is the comparison of the new algorithm
to the incorrect original algorithm and on the right scatter plot
we compare to the fixed original algorithm.
As we can see, the new algorithm outperforms the original
algorithm significantly.
Since we did not prove that the fixed original algorithm is correct,
we at least checked the outputs of the new NTSCD and the fixed original
algorithms and found out that they output precisely the same control dependence
relations.
We can also see that the scatter plot on the right contains more timeouts
of the original algorithm. It supports the claim that the fix
slows down the algorithm.


\subsubsection{DOD algorithms}
We compared the new DOD algorithm to the original DOD algorithm with
the fix described in Subsection~\ref{ssec:dod_ranganath}.
As the fix does not change the complexity of the original algorithm,
we do not compare the new algorithm with the incorrect version of the original
algorithm.
The results of the experiments are displayed in Figure~\ref{fig:dod_experiments}
on the left.  We can see that the new algorithm is again very fast.
In fact, the results resemble the results of the pure NTSCD algorithm
(which is more or less a part of the DOD algorithm in the form of computing
$V_p$ sets). Our algorithm takes the advantage that it can answer
early that there are no DOD dependencies for a predicate node.
Although this fact does not change the worst-case complexity,
it helps in practice.

\begin{figure}[t]
\begin{tabular}{c p{1mm} c}
\hspace*{-3mm}
\includegraphics[width=6.5cm]{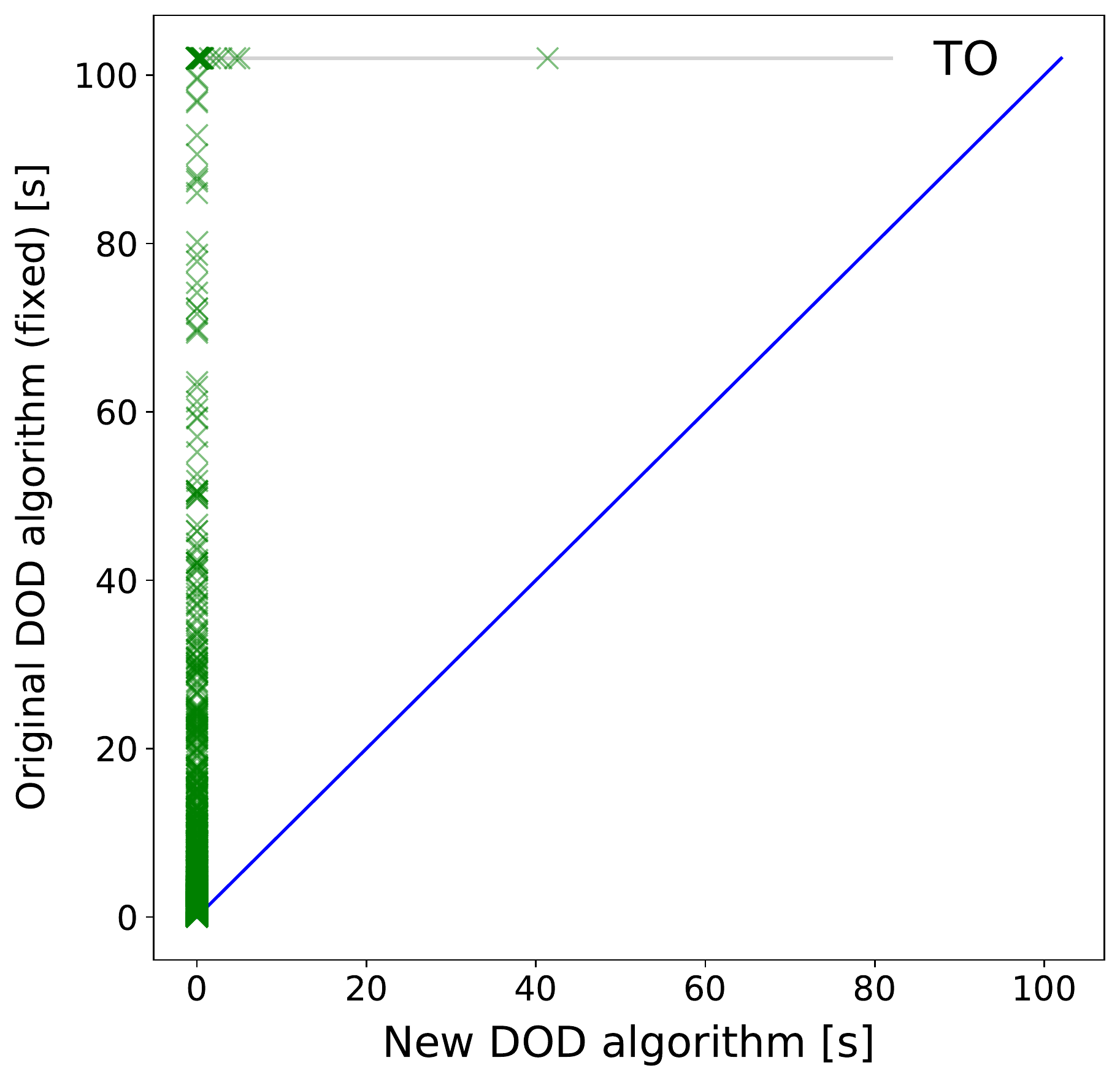}
&&
\includegraphics[width=6.5cm]{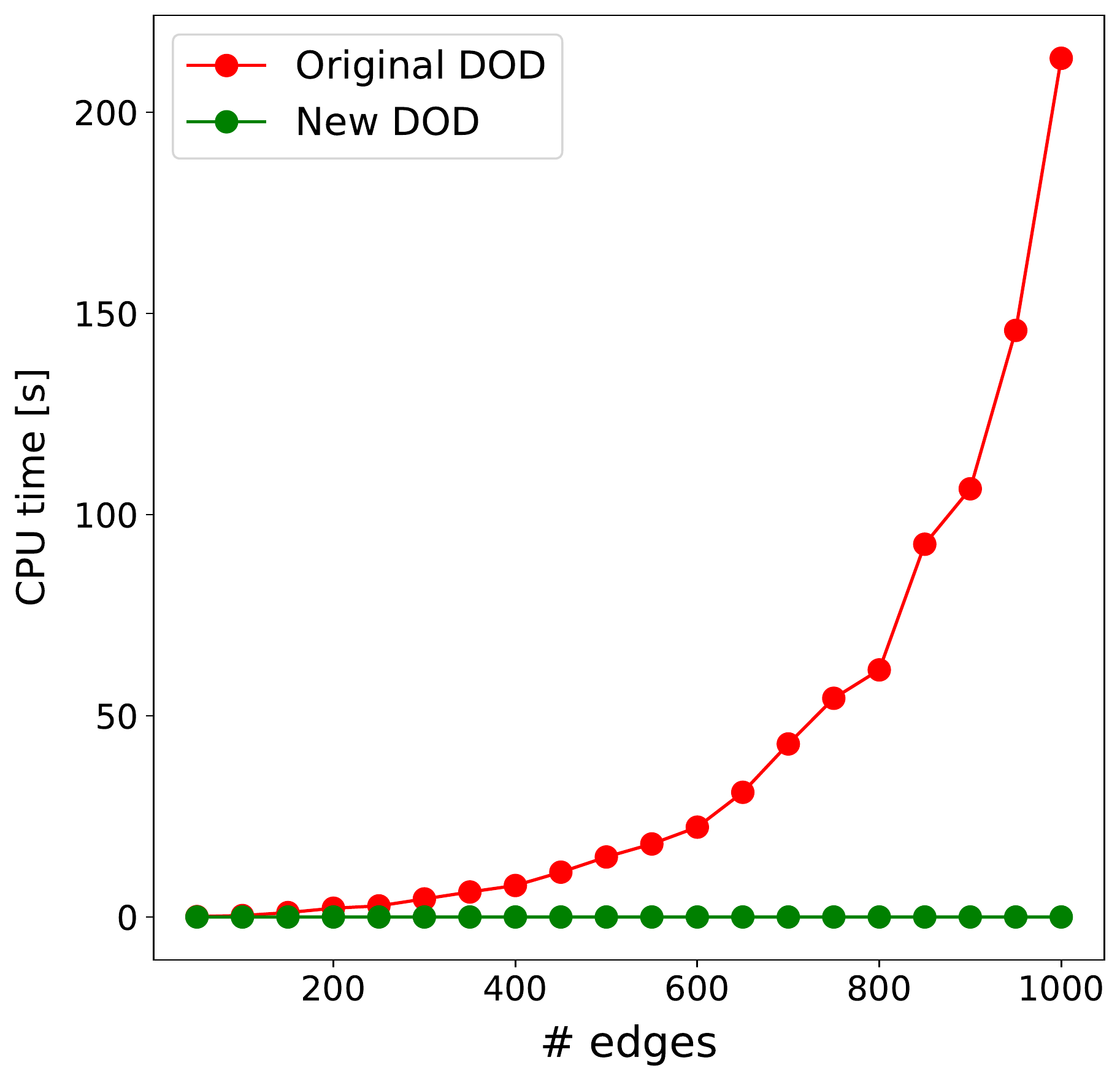}
\end{tabular}
\caption{On the left are compared the CPU running times of the original (fixed)
         and the new DOD algorithms. On the right, the plot shows the running
         time of the algorithms on random graphs with 500 nodes and the number
         of edges as specified by the $x$-axis. All CPU time is in seconds.}
\label{fig:dod_experiments}
\end{figure}

Irreducible graphs are rather rare in common code
(because structured programs yield reducible graphs).
Therefore, we tested the DOD algorithms also on randomly generated graphs,
where we can expect that an irreducible graph emerges more often.
Figure~\ref{fig:dod_experiments} on the right shows the results for graphs that
have 500 nodes and 50, 100, 150, \dots randomly distributed edges
(such that every node has at most two successors).
The running CPU time is an average of 10 measurements.

We can see that the new algorithm is agnostic to the number of edges.
Its running time in this experiment ranges from
$4.12\cdot 10^{-3}$ to $8.89\cdot 10^{-3}$ seconds.
The original DOD algorithm does not scale well
with the increasing number of edges.

\subsubsection{Comparison to control closures}

As computing strong control closure may be equivalent to computing
NTSCD and DOD (see Section~\ref{sec:cc}), we conducted a set of experiments
also with strong control closures.
One obstacle in these experiments was that control closures need a starting
set of nodes that is going to be closed.
We decided to run these experiments only on the subset of
functions that have at least two exit points and set as the starting set
the entry point and one of the exit points.
This approach makes the closure algorithm to compute which nodes may influence
getting to the other exit node.
The new algorithm for the computation of DOD and NTSCD was setup to compute
all the control dependencies for the given function and then
construct the closure of the same set of starting nodes.

\begin{figure}
\includegraphics[width=6.5cm]{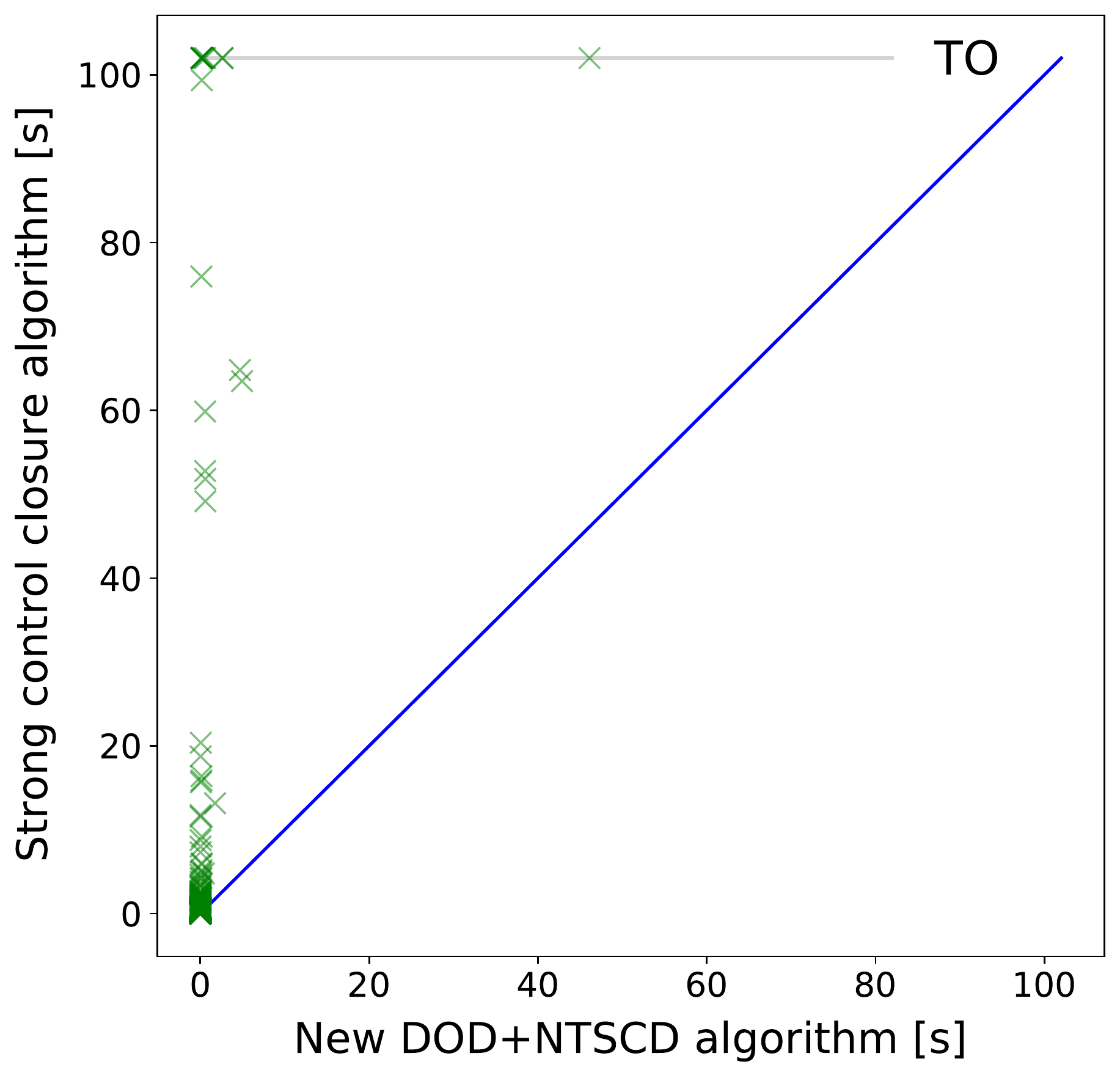}
\caption{The comparison of the CPU running time of computing strong control
         closure using the \citet{Danicic11} algorithm and using
         the new DOD + NTSCD algorithm.}
\label{fig:scc}
\end{figure}

Results are shown on the scatter plot in Figure~\ref{fig:scc}.
As expected, our algorithm runs similarly fast as in the other experiments
and it outperforms the algorithm for the computation of control closures
that has worse complexity.
An observant reader can notice that in this set of experiments, the new
algorithm has no result around 15 seconds even when the pure NTSCD
computation has such results. It turns out that the computation
of $A_p$ sets (which is basically the new NTSCD algorithm)
that gathers the marks into bitvectors is in practice faster than
the NTSCD algorithm (including deciding the NTSCD relation).
The main reason for this is that the implementation with bitvectors
does not need to initialize the counters for each call of $\compute$.



\section{Conclusion}
\label{sec:conclusion}

\begin{table}[tb]
  \caption{Overview of discussed algorithms and their complexities on a CFG $(V,E)$}
  \centering
  \begin{tabular}{lll}
  \toprule  
    relation/closure & algorithm & complexity\\ \midrule
    NTSCD     & original algorithm by \citet{Ranganath07} & $O(|V|^4\cdot \log|V|)$ \\
              & fixed original algorithm by \citet{Ranganath07} & $O(|V|^5)$ \\
              & new algorithm & $O(|V|^2)$ \\ \midrule
    DOD       & original algorithm by \citet{Ranganath07} & $O(|V|^5\cdot \log|V|)$ \\
              & fixed original algorithm by \citet{Ranganath07} & $O(|V|^5\cdot \log|V|)$ \\
              & new algorithm & $O(|V|^3)$ \\ \midrule
    strong CC & original algorithm by \citet{Danicic11} & $O(|V|^4)$ \\
              & new NTSCD+DOD algorithm and computation of closure & $O(|V|^3)$ \\
    \bottomrule  
  \end{tabular}
  \label{tab:all}
\end{table}

We studied algorithms for the computations of strong control
dependencies, namely non-termination sensitive control dependence
(NTSCD) and decisive order dependence (DOD) by \citet{Ranganath07} and
strong control closures (CC) by \citet{Danicic11} on arbitrary control
flow graphs where each branching statement has two successors. We have
demonstrated flaws in the original algorithms for computation of NTSCD
and DOD and we have suggested corrections. Moreover, we have
introduced new algorithms for NTSCD, DOD, and CC that are
asymptotically faster (see Table~\ref{tab:all}).  All the mentioned
algorithms have been implemented and our expriments confirm better
performance of the new algorithms.

\begin{acks}                            
  This material is based upon work supported by
  the Czech Science Foundation grant GA18-02177S.

\end{acks}


\bibliography{references}

\end{document}